\documentclass[11pt]{article}
\usepackage[margin=1in]{geometry} 
\usepackage{textcomp}
\usepackage{hyperref}
\usepackage{gensymb}
\usepackage{graphicx}
\usepackage{amssymb,amsmath}
\usepackage{amstext,amsthm, amssymb,amsfonts}
\usepackage{epstopdf}
\usepackage[usenames,dvipsnames]{xcolor}
\usepackage{multirow}
\usepackage{placeins}
\usepackage{soul}
\usepackage{amssymb,amsfonts,amsmath,amsthm}
\usepackage{algorithm,algorithmic}

\theoremstyle{lemma}
\newtheorem{lemma}{Lemma}
\newtheorem{theorem}{Theorem}
\newtheorem{definition}{Definition}

\title{Configuring Random Graph Models with Fixed Degree Sequences\thanks{All authors contributed equally to this work.}}

\author{
  Bailey K. Fosdick\thanks{Department of Statistics, Colorado State University, Ft. Collins, CO 80523 USA (bailey@stat.colostate.edu)}
  \and
  Daniel B. Larremore\thanks{Santa Fe Institute, 1399 Hyde Park Rd. Sante Fe, NM, 87501 USA; Department of Computer Science, University of Colorado, Boulder, CO 80309, USA; BioFrontiers Institute, University of Colorado, Boulder, CO 80303, USA. (daniel.larremore@colorado.edu)}
  \and
  Joel Nishimura\thanks{School of Mathematical and Natural Sciences, Arizona State University, Glendale, AZ 85306 USA (joel.nishimura@asu.edu)}
  \and
  Johan Ugander\thanks{Management Science \& Engineering, Stanford University, Stanford, CA, 94305 USA (jugander@stanford.edu)}
}

\begin{document}
\maketitle

\begin{abstract}
Random graph null models have found widespread application in diverse research communities analyzing network datasets, including social, information, and economic networks, as well as food webs, protein-protein interactions, and neuronal networks. The most popular family of random graph null models, called configuration models, are defined as uniform distributions over a space of graphs with a fixed degree sequence. Commonly, properties of an  empirical network are compared to properties of an ensemble of graphs from a configuration model in order to quantify whether empirical network properties are meaningful or whether they are instead a common consequence of the particular degree sequence. In this work we study the subtle but important decisions underlying the specification of a configuration model, and investigate the role these choices play in graph sampling procedures and a suite of applications. We place particular emphasis on the importance of specifying the appropriate graph labeling---stub-labeled or vertex-labeled---under which to consider a null model, a choice that closely connects the study of random graphs to the study of random contingency tables. We show that the choice of graph labeling is inconsequential for studies of simple graphs, but can have a significant impact on analyses of multigraphs or graphs with self-loops. The importance of these choices is demonstrated through a series of three in-depth vignettes, analyzing three different network datasets under many different configuration models and observing substantial differences in study conclusions under different models. We argue that in each case, only one of the possible configuration models is appropriate. While our work focuses on undirected static networks, it aims to guide the study of directed networks, dynamic networks, and all other network contexts that are suitably studied through the lens of random graph null models.
\end{abstract}

\section{Introduction}\label{sec:introduction}

A configuration model is a uniform distribution over graphs with a specific degree sequence. For researchers studying network data, it is common to employ a configuration model as a degree-preserving null model that holds fixed the degree sequence of an empirical graph while randomizing all other structure. In other domains, researchers study the properties of graph algorithms, dynamical models, or optimization routines on ``realistic" graphs by sampling random graphs from a configuration model with an empirically relevant degree sequence.

There is a tendency in the literatures of graph mining, machine learning, and network science to think of and study {\it one} configuration model---{\it the} configuration model---without specifying or reflecting upon the defining properties of the space of graphs over which the uniform distribution is considered. As a consequence, misunderstandings have developed within a number of domain sciences surrounding {\it the} configuration model, at times because discussions refer to uniform distributions over subtly but importantly different spaces of graphs. In this paper, we clarify the differences between eight commonly arising graph spaces and their corresponding uniform distributions, aiming to provide an orderly review and guide for the diverse fields of study where configuration models have found application.

In some circumstances, differences between particular graph spaces are asymptotically small in the limit of large and sparse graphs with restricted degree sequences. However, as we will demonstrate, not all differences between graph spaces are asymptotically small, and perhaps more importantly, a great deal of modern graph analysis is performed on graphs that are well short of fulfilling these asymptotic promises. 

We begin by reviewing eight common graph spaces over which one might seek a uniform distribution. These spaces can be organized according to the answers to three binary questions, which we describe in Section \ref{subsec:choosing}. We then provide a detailed overview of the subtleties involved in uniformly sampling from these different spaces in Sections \ref{sec:sampling} and \ref{sec:othersampling}, primarily through correctly specified Markov chains, with brief discussions of other related graph spaces, including connected, directed, and weighted graphs\footnote{See also  \cite{carstens2016switching} whose publication followed this work's submission.}. After establishing formal sampling results we then turn to a series of three vignettes in Section \ref{sec:vignettes} that illustrate the scientific importance of choosing the correct graph space as a null model. In particular, we argue that the common default choice of studying configuration models over stub-labeled graphs (where each half-edge is labeled) is an inappropriate choice for most analyses of non-simple graphs. Importantly, we demonstrate that this choice of null model leads to different conclusions than more appropriate null models based on vertex-labeled graphs.

\subsection{Basic definitions}
Recall the basic definition of a {\it graph} as an ordered pair $G=(V,E)$, consisting of a vertex set $V$ and an edge set $E \subseteq V \times V$. The edge set $E$ is understood to be a simple set, but if $E$ is a multiset (where a vertex pair $(u,v)$ can appear several times in $E$) then the graph is instead called a {\it multigraph}. Depending on the context, a graph or multigraph may allow or disallow the presence of self-loops (edges of the form $(u,u)$, connecting a vertex to itself). A graph is also often represented as a $|V| \times |V|$ \textit{adjacency matrix}, such that the $(i,j)$th entry $w_{ij}$ is equal to the number of edges between vertices $i$ and $j$. For undirected graphs, as considered here, the adjacency matrix is symmetric. 

The choices to allow or disallow self-loops or multiedges are the first two choices in specifying a configuration model's graph space. In order to be precise about the properties of each graph space, we briefly review four definitions. First, a {\it simple graph} is a graph without self-loops or multiedges. Second, there is no established name in the literature for a graph allowing self-loops but without multiedges, so we refer to such a graph plainly as a {\it loopy graph}. In the literature, multigraphs are sometimes taken to have self-loops or not; we adopt the more conventional name {\it multigraph} to refer specifically to multigraphs without self-loops, and use {\it loopy multigraph} to refer to a multigraph that allows self-loops (also sometimes called a pseudograph). See Figure \ref{fig1}(a) for a diagram illustrating the basic relationships between these graph spaces.

\subsection{Vertex- and stub-labeled graph spaces}

A graph $G=(V,E)$ consists of two sets: a vertex set $V$ and an edge set $E$. These sets can be unlabeled or labeled, motivating the following  definitions that will be used throughout the paper. 
\begin{definition}[Vertex-labeled graph]
	A \underline{vertex-labeled graph} is a graph in which each vertex has a distinct label.
\end{definition}
For vertex-labeled graphs, there is a bijection between graphs and adjacency matrices, i.e.~each vertex-labeled graph can be uniquely identified by its adjacency matrix and vice versa. However, in addition to vertices, the two endpoints of each edge (where they connect to vertices), can also be labeled separately. The case when these half-edges or ``stubs'' are labeled motivates the following definition.
\begin{definition}[Stub-labeled graph]
	A \underline{stub-labeled graph} is a graph in which each half-edge ({\it stub}) has a distinct label, and thus each edge has a pair of distinct labels.
\end{definition}
Note that a stub-labeled graph also has implicitly labeled vertices, since each vertex is distinctly labeled by the set of labeled stubs attached to it.  However, in contrast with vertex-labeled graphs, there is not a bijection between stub-labeled graphs and adjacency matrices, i.e.~multiple stub-labeled graphs can correspond to the same adjacency matrix. An {\it unlabeled graph} is a graph in which neither edges nor vertices are labeled. An unlabeled graph can be thought of as an isomorphism class in a space of labeled graphs, where there exists a set of labeled graphs that all correspond to the same unlabeled graph. Similarly, there exists a set of stub-labeled graphs which correspond to the same vertex-labeled graph, motivating the following definition.
\begin{definition}[Stub-isormorphism]
	A \underline{stub-isomorphism equivalence class} is the set of all stub-labeled graphs which, upon removal of stub labels, results in the same vertex-labeled graph. Equivalently, a stub-isomorphism class is the set of all stub-labeled graphs which are represented by the same adjacency matrix. Two graphs in the same stub-isomorphism class are said to be \underline{stub-isomorphic}.
\end{definition}

For the space of simple graphs with a given degree sequence $\{k_i\}_{i \in V}$, where $k_i$ is the degree of vertex $i$---and only for simple graphs, as we shall see---the number of stub-isomorphic graphs corresponding to a given vertex-labeled graph is a constant that depends only on the degree sequence (which is fixed). As a result, each vertex-labeled graph appears the same number of times in the space of stub-labeled graphs, and hence, the uniform distributions over both spaces are equivalent in most practical contexts where analyses ignore explicit stub labels. On the other hand, for non-simple graphs with loops and/or multiedges, this is not the case, and the choice of labeling can radically change the space of graphs, and thereby, a resulting/downstream/derivative analysis.

We visualize the differently labeled spaces for an example degree sequence, $\{2,2,1,1\}$, in Figure \ref{fig1}(c-e). In the vertex-labeled space, half the graphs (3 of 6) have self-loops and only a third of the graphs (2 of 6) are simple; in the stub-labeled space, the majority of the graphs (8 of 15) are simple. As we will show in Section \ref{sec:enumeration}, self-loops and multiedges are always more common in vertex-labeled graphs, and for many degree sequences they are vastly more common. Uniform distributions over these differently labeled spaces can therefore produce wildly different answers to straightforward questions. For example, if one asks, ``What fraction of graphs with the given degree sequence form a single connected component?''for this degree sequence, the answer varies considerably---1/4, 2/6, or 8/15---depending on the space.

\begin{figure}[t]
	\centering
	\includegraphics[width=0.95\linewidth]{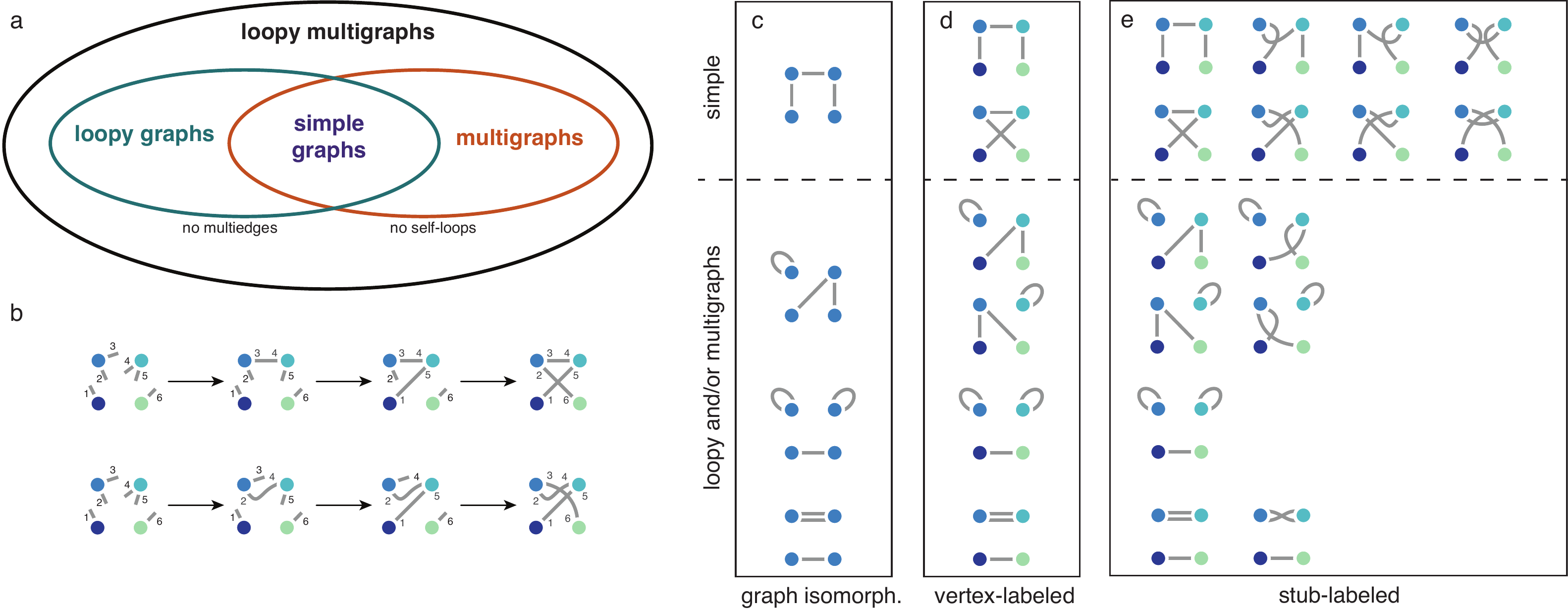}
	\caption{{\bf Graph spaces.} (a) Nested and overlapping graph spaces, defined by allowing or prohibiting self loops or multiedges. (b) Two instances of stub-matching resulting in the same vertex-labeled graph but different stub-labeled graphs. (c-e) For the degree sequence $\{k_i\}=\{2,2,1,1\}$, the (c) set of graph isomorphism classes, (d) set of vertex-labeled graphs, and (e) set of stub-labeled graphs, (where the stub labels are delineated by the locations where they protrude from a vertex). For the two simple graphs in panel (d), they are both ``stub-isomorphic'' to the same number of stub-labeled graphs in panel (e), in particular, to exactly $\prod_{i} k_i ! = 4$ graphs. However, the sizes of the stub-isomorphism classes differ for graphs with self-loops or multiedges, illustrating why vertex- and stub-labeled spaces may not be treated as equivalent. Note that both graphs shown in panel (b) fall in the same row of panel (e).}
	\label{fig1}
\end{figure}

\subsection{A brief history of stubs}
\label{sec:stub}

Stub-labeled graphs arise naturally from a relatively simple stub matching process. The first step 
assigns a specific number of stubs to each vertex, ensuring that each vertex will have exactly the desired number of edges as specified by the degree sequence. To guarantee vertex $i$ will have the correct degree $k_i$, we force one endpoint of each of $k_i$ edges to be vertex $i$ while the other endpoint is left floating, unassigned. In this way, each vertex $i$ has $k_i$ half-edges or stubs. Joining two such stubs produces an edge. Note that by construction, every vertex has the correct number of edges, so repeatedly joining pairs of stubs results in a graph with the correct degree sequence, shown in Fig.~\ref{fig1}(b).

More precisely, the stub matching process takes a specified degree sequence $\{k_i\}_{i \in V}$ and generates a graph using the following randomized process. Each vertex $i$ is assigned exactly $k_i$ stubs, and pairs of stubs are chosen uniformly at random and connected until there are no remaining unpaired stubs. This process, which only requires that the total number of stubs be even, creates a loopy multigraph with exactly the specified degree sequence. Due to the fact that stubs are chosen uniformly at random, this {\it stub-matching} procedure (also called the {\it pairing model} \cite{blitzstein2011sequential}) samples uniformly from the space of stub-labeled loopy multigraphs, as discussed further in Section~\ref{sec:directsample}.

Stub matching was first introduced by Bollob\'as \cite{bollobas1980probabilistic} as a method for enumerating the number of vertex-labeled simple graphs with certain degree sequences \cite{bekessy1972asymptotic,bender1978asymptotic}. Although stub matching draws from the space of stub-labeled loopy multigraphs, Bollob\'as assumed that the degrees of all vertices did not grow too quickly, relative to the size of the graph, and then showed that the number of stub-labeled graphs with self-loops and/or multiedges was asymptotically small relative to the number of stub-labeled simple graphs. By the fact that every vertex-labeled simple graph is stub-isomorphic to exactly $\prod_{i} k_i!$ stub-labeled graphs (see Section \ref{sec:enumeration} and Figure~\ref{fig1}(d-e)), Bollob\'as provided an asymptotically tight estimate (for large graphs) of the number of vertex-labeled simple graphs. Of note, Bollob\'as called each stub-labeled graph a {\it configuration}, the origin of the name {\it configuration model} for these uniform distributions.

Bollob\'as' analysis contains two subtleties that are major sources of confusion about configuration models. First, as noted above, every vertex-labeled simple graph is isomorphic to a fixed number of stub-labeled simple graphs (e.g.~this number is four for the degree sequence $\{2,2,1,1\}$ in Figure \ref{fig1}), but the same cannot be said for graphs with self-loops or multiedges. Second, many analyses assume conditions on the degree sequence (e.g., adequately bounded growth) under which the number of non-simple graphs is asymptotically small relative the number of simple graphs, but for any finite degree sequence the number of non-simple graphs can represent a substantial fraction of the graph space. The mathematical literature is almost always precise regarding these two points. However, as configuration model random graphs have spread into diverse fields due to waves of interest in graph analysis and network science methods, these points have often caused confusion in the broader literature, as we discuss below. We hope that this work helps mark a turning point in that confusion. In the remainder of this introduction, we briefly survey the history of different applications of fixed-degree-sequence random graph null models, and then summarize the concrete decisions that underlie the choices of different configuration model null models.

\subsection{A brief history of applications of random graphs with fixed degree sequence} \label{subsec:history}

The practice of comparing an observation to a randomized null model has its origins in R.~A.~Fisher's foundational work on randomization for hypothesis testing \cite{fisher1935design}.
Random graph null models extend this practice to the space of graphs. They allow comparisons between properties of real-world graphs and properties of graphs drawn at random from a graph space, ultimately allowing us to quantify what is surprising and what is expected. However, as with any hypothesis test, the choice of randomized null model directly affects the conclusions that can be drawn from the test. For this reason, the classic but overly simplistic Erd\H{o}s-R\'enyi random graph model, in which each possible edge exists independently with probability $p$, or its near equivalent, in which a fixed number of edges are placed between random pairs of vertices, are usually avoided. Compared to an Erd\H{o}s-R\'enyi null model, 
 real-world networks often appear rich in structure by comparison. Instead, due to the fact that many key properties of networks are strongly constrained by the distribution of vertex degrees \cite{newman2001random,boguna2004cut,callaway2000network,cohen2002percolation,larremore2011predicting,restrepo2005onset}, it is far more common and appropriate to use as a null model a space of graphs in which the degrees of all the vertices are fixed, but where the edges are otherwise placed between vertices uniformly at random. This family of degree-preserving random graph models, which we call configuration models throughout this paper, have been discovered independently and used as null models in sociology, ecology, systems biology, combinatorics, statistics, psychology and network science, spanning over 80 years of applied research. We detail some of this rich history here.

{\it Null models in sociology: chance sociograms, 1930s}. In 1934 Jacob Moreno initiated the quantitative study of social networks through his influential book {\it Who Shall Survive?} \cite{moreno1934shall}. Soon thereafter, in 1938, Moreno and Jennings published {\it Statistics of Social Configurations}, which introduced statistics to social network analysis through the use of so-called {\it chance sociograms}, i.e.~randomly sampled adjacency matrices with fixed out-degrees (i.e.~one fixed margin) \cite{moreno1938statistics}. Moreno and Jennings argued that in order to establish the statistical significance of an analysis, one should compare an observed social network with a network constructed through a chance experiment.\footnote{Moreno and Jennings in fact frequently used the word ``configurations'' to describe their chance sociograms, several decades before Bollob\'as' work: ``Study of the findings of sociometric tests showed that the resulting configurations, in order to be compared with one another, were in need of some common reference base from which to measure the deviations. It appeared that the most logical ground for establishing such a reference could be secured by ascertaining the characteristics of typical configurations produced by chance balloting for a similar size population with a like number of choices.'' That said, the term configuration model is generally accepted to stem from Bollob\'as' usage of the word.} Moreno and Jennings demonstrated their procedure by studying a population of 26 children at the New York State Training School for Girls in Hudson, NY. The children were surveyed for their three preferred dining partners, creating a directed network of dining partner preferences. This observed network was compared to a small set of seven manually randomized directed graphs restricted such that each vertex had three outgoing edges and no multiedges (as in the observed network). Moreno and Jennings contrasted their empirical graph with their small ensemble of graphs drawn from their null model, and concluded that some observed network features were statistically significant while others were not. While our focus in this work is on undirected (as opposed to directed) configuration models, directed configuration models are discussed briefly in Section~\ref{sec:othermarkov}.
Another significant early use of a random graph null model in sociology is contained in Davis and Leinhardt's work testing Homans' structural theory of social hierarchy from the 1950's \cite{homans1950human}. The study tested the theory by studying social network subgraph frequencies \cite{davis1971structure}, contrasting empirical frequencies with those of an Erd\H{o}s-R\'enyi random graph null model.

{\it Null models in ecology: species co-occurrence patterns, 1970s.} A configuration model arose independently in ecology when, in 1975, Jared Diamond published an analysis of bird species co-occurrence on the islands of the Bismarck Archipelago and argued that, based on the patterns of species presence and absence observed across the islands, the presence of some species precluded the presence of others \cite{diamond1975assembly}. In 1979, Connor and Simberloff argued that the patterns themselves were not sufficient evidence for such conclusions; they argued that a null model of randomly assigned species to islands, in which the number of species per island and number of islands per species are exactly preserved, should be used to assess the possibility that the empirical patterns are the result of random chance \cite{connor1979assembly}. In other words, Connor and Simberloff argued that observed patterns should be compared against a null model, and in particular against a degree-preserving configuration model, based on the observed presence/absence matrix. This methodological debate has continued for over 40 years regarding both the correct null model and appropriate test statistics for quantifying patterns of species presence/absence patterns (see \cite{gotelli1996null} for a partial review). 

Many contributions to the ongoing ecological discussion have been made in the years since. In 1987, Wilson contributed a fixed marginal null model, which required that any matrix in the ensemble have the same number of sites per species and species per site as the observed data, corresponding directly to an undirected bipartite configuration model with fixed degrees \cite{wilson1987methods}.\footnote{A bipartite network is a network where edges only occur between two distinct sets of vertices. For example, a plant-pollinator network contains both plants and insects as vertices and edges connecting pollinating insects to plants, and no edges between pairs of insects or pairs of plants.} Wilson's 1987 fixed marginal null model assembled the network via a stub matching procedure. He found that often, the stub matching was unable to finish without creating a double edge, and so he found better success rates by using a heuristic nearly equivalent to the Havel-Hakimi algorithm \cite{havel1955remark,hakimi1963realizability} (though Wilson states that he was unable to find any proof in the literature of his method). This debate illustrates the disconnect between the ecology and mathematics literatures at the time.

{\it Null models for tables: matrix counting and contingency tables, 1970s - 1990s.} Contingency tables are rectangular matrices with integer entries, representing a tabulation of entities along two dimensions, e.g.~the number of college graduates by major and institution. These tables, when viewed as adjacency matrices, characterize an undirected bipartite multigraph. There are straightforward analogous connections between the binary tables in ecology and the more general (non-binary) contingency tables studied in statistics \cite{chen2005sequential}. As in the network literature, contingency table analyses often involve asking whether table properties are interesting compared to random tables with the same row and column totals (the same marginal totals). An initial focus of this literature was on enumerating the number of matrices with fixed marginals \cite{gail1977counting,diaconis1995rectangular}. Compared to presence/absence matrices, where the entries are restricted to be either 0 or 1, analyzing adjacency matrices corresponding to contingency tables is much more straightforward. Many direct sampling procedures have been proposed \cite{patefield1981algorithm}, as well as procedures which \textit{exactly} characterize the null distribution of tables with fixed marginals and do not rely on sampling (see \cite{verbeek1985survey,agresti1992survey} for reviews of these methods).

{\it Null models in systems biology: network motifs, 2000s.} As the large-scale study of both genetic regulatory networks and neuronal networks emerged in the early 2000s, lengthy debates were held in the literature regarding the choices of (and technical means for sampling from) null models. The debate on genetic regulatory networks began with a study by Milo {\it et al.}~that found specific network motifs (regulatory patterns) that were more frequent than expected in a configuration model null model \cite{milo2002network,itzkovitz2003subgraphs}. Soon after that work was published, King issued a commentary that called attention to choices in the design of the random graph sampling algorithms in these works, noting that they did not sample uniformly from any graph spaces of reasonable interest \cite{king2004comment}. A series of responses by the original authors led to corrected algorithms for sampling from the stub-labeled spaces of random graphs with fixed degree sequences \cite{milo2003uniform, itzkovitz2004reply}. It is worth mentioning that other work on configuration model null models of genetic regulatory networks, using correct sampling techniques, was also being conducted in parallel to the above controversy \cite{maslov2002specificity}.

A parallel debate in the literature on neuronal networks noted that the study of network motifs in neuronal networks \cite{milo2002network,milo2004superfamilies} involving comparisons between observed structures and configuration model random graphs was flawed at a deeper conceptual level, as it overlooked the role of spatial structure in brains \cite{artzy2004comment}. A series of published exchanges followed \cite{milo2004response,artzy2005generating}, leading to the study of specific spatial network null models for studying brain networks \cite{sporns2004motifs}. A similar adaptation, known as distance modularity \cite{liu2014detecting}, has recently been introduced to the broader literature on network community detection.

Other applications of configuration model random graph null models include studies of patterns in the structure of the world wide web \cite{newman2001random}, the Internet \cite{maslov2004detection}, food webs \cite{stouffer2007evidence}, academic career trajectories \cite{malmgren2010role}, the dynamics of social contagion \cite{centola2007cascade}, disease propagation \cite{st2017susceptible}, opinion dynamics \cite{watts2007influentials}, and economic network effects \cite{sundararajan2007local}. As we discuss at length in Section~\ref{sec:modularity}, these null models also underlie all community detection methods based on modularity maximization \cite{newman2004finding}. Across these diverse applications as well as the earlier literatures, different applications have tended to employ slightly different null models, and these variations make it very difficult to compare and contrast findings. In the next subsection we introduce a sequence of concrete choices that formalize the decisions underlying the choice of a graph space, and hence a configuration model. Consequences of these decisions are discussed at length in Section~\ref{sec:vignettes} through a series of application vignettes.

\subsection{Choosing a graph space}\label{subsec:choosing}
It is often impossible to unambiguously identify an empirical graph as coming from a particular space of graphs; additional knowledge about the system that produced the graph is almost always required. For example, as shown in Figure~\ref{fig1}, simple graphs are a subset of the other graph spaces, and thus a given simple graph may plausibly lie within any of the spaces, defined by the presence or absence of self-loops, multiedges, and stub-labels. Therefore, in order to choose the appropriate graph space for a null model, we introduce three questions about the graph and the system that produced it. \\

\begin{figure}[t]
	\centering
	\includegraphics[width=0.75\linewidth]{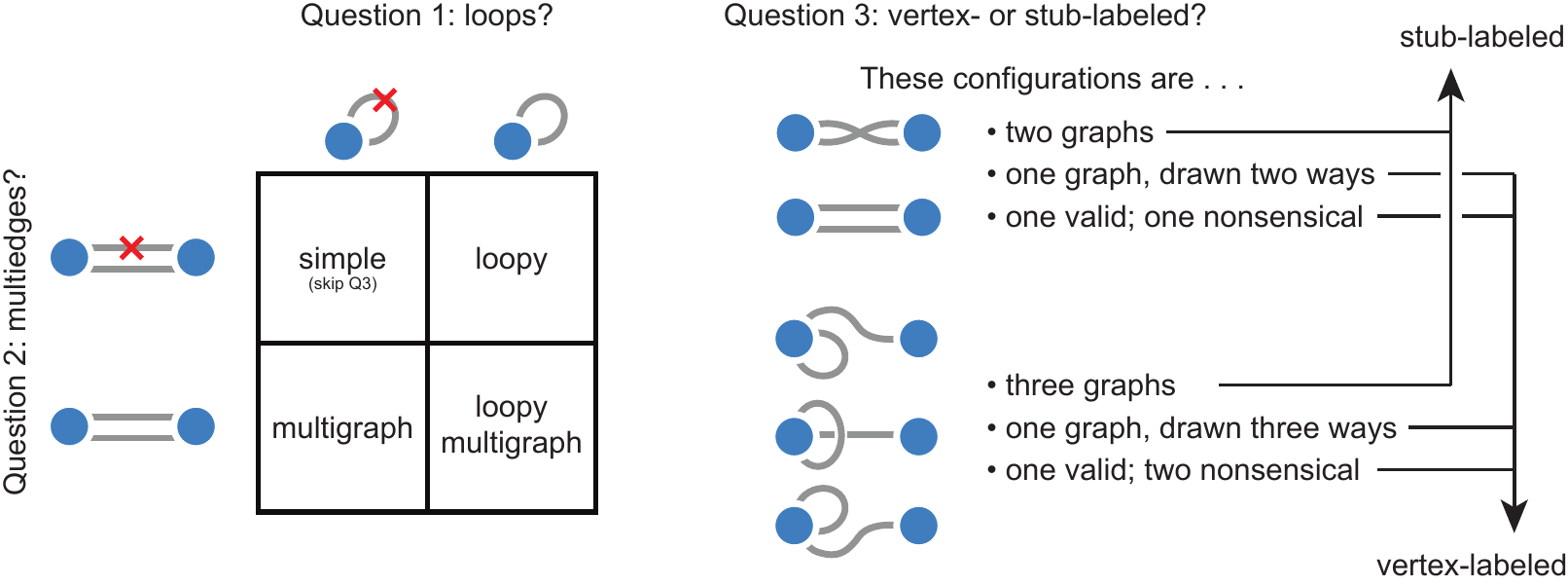}
	\caption{{\bf Choosing a graph space.} Three questions must be answered in order to correctly choose a configuration model graph space. Questions 1 and 2 address whether the graph has, or could possibly have, self-loops and multiedges. If the space permits self-loops, multiedges, or both, then Question 3 addresses whether the space is vertex-labeled or stub-labeled. These questions are explained in detail in the text, Section~\ref{subsec:choosing}.}
	\label{fig2}
\end{figure}

{\bf Question 1: Are there self-loops in the graph?} 
For example, a citation network consisting of papers (as vertices) and their citation relationships (as edges) cannot have self-loops since a single paper can never cite itself. On the other hand, a network of authors (as vertices) and their citation relationships (as edges) may very well have self-loops since authors can, and do, cite their own work. Note that an observed network of authors and their citations ought to reside within a graph space allowing self-loops, even if a particular observed network has no self-loops. However, in some cases, the method of data collection or recording itself may preclude self-loops---even if a self-loop would be reasonable and interpretable---and in such cases, the relevant graph space should not include self-loops.

{\bf Question 2: Are there multiedges in the graph?}
For example, a network of contacts among barn swallows---analyzed in Section~\ref{sec:swallows}---in which each edge represents an observed interaction between a pair of birds, may have multiedges corresponding to multiple observations of an interaction between the same pair of birds. On the other hand, a protein-protein interaction network, in which two proteins are connected if they interact, cannot ever have a multiedge since interactions in this context are conceptually boolean. Note that an observed network may reside within a graph space allowing multiedges, even if a particular observed network has no multiedges. However, as in Question 1, in some cases, the method of data collection or recording itself may preclude multiedges---even if a multiedge would be reasonable and interpretable---and in such cases, the relevant graph space should not include multiedges.

If the answers to the first two questions are both no, then the space of simple graphs is the appropriate space. For the purposes of sampling from a simple configuration model, there is then no meaningful difference between vertex- and stub-labeled spaces. One need only to ensure that the graph sampling algorithm correctly samples from the space of simple graphs (a non-trivial task further discussed in Section \ref{sec:sampling}), due to the fact that any ensemble of vertex-labeled simple graphs can easily be converted into an ensemble of stub-labeled simple graphs, and vice versa (see Section \ref{sec:enumeration} for further discussion).  However, if the answer to either of the previous questions was yes, indicating that the graph space contains self-loops, multiedges, or both, we pose a key third question.

{\bf Question 3: Is the graph space stub-labeled or vertex-labeled?} Consider a pair of vertices connected by two edges. If swapping the edges so that they cross, as shown in Figure~\ref{fig2}, produces a distinct graph, the space is stub-labeled. Alternatively, if crossing the edges either produces a graph with the same interpretation or produces a nonsensical graph, the space is vertex-labeled.

There are a number of instances where a graph should be treated as vertex-labeled rather than stub-labeled. For example, if the stubs are ordered (e.g.~temporally) in a way that would make swapping nonsensical, the space of graphs is vertex-labeled in spite of the fact that the stubs have identities. Such a situation is commonly encountered when studying a telephone network (also called a call detail record or CDR), where edges represent phone calls between individuals. If a pair of individuals are recorded sharing two phone calls, it is meaningless to consider the crossed graph that connects the stub associated with the first call and the first individual to the stub associated with the second call and the second individual, as this swap represents a graph that could never have been observed. See Section~\ref{sec:swallows} for a concrete exploration of these differences. If, on the other hand, the crossed edges and parallel edges as shown in Figure~\ref{fig2} are distinguishable and plausible, the space of graphs should be stub-labeled. For example, in a network of intermarriages between families or villages, an edge may correspond to an individual from one village marrying an individual from another village. Here, different sets of marital pairings are meaningful and distinct, indicating that the graph space is stub-labeled. 

One alternative approach to answering Question 3 involves considering the adjacency matrix of the graph. For a vertex-labeled space, each graph corresponds to a single, unique adjacency matrix, and each adjacency matrix corresponds to a single, unique vertex-labeled graph. On the other hand, multiple stub-labeled graphs have identical adjacency matrices, and a valid adjacency matrix corresponds to a stub-isomorphism class of stub-labeled graphs, as shown in Figure~\ref{fig1}. Thus, Question 3 may be answered by considering whether the adjacency matrices corresponding to the graph space are unique and distinct objects, or whether repeated adjacency matrices are allowed in the ensemble.\\

Answers to the first two questions in this section fully specify whether the graph space is simple, loopy, multigraph, or loopy multigraph, and the answer to the third question determines whether the space is stub-labeled or vertex-labeled. Since, for the purposes of sampling simple graphs or analyzing network properties that are functions of the adjacency matrix, there is no practical difference stub-labeled and vertex-labeled spaces, we may often treat these as equivalent and focus on the seven distinct and non-interchangeable spaces of graphs just described. 
\\

{\bf Organization.}
In Section~\ref{sec:sampling} we describe space-specific Markov chain Monte Carlo algorithms that provably generate uniform samples from the graph spaces discussed above. Alternative methods for sampling random graph null models are discussed in Section~\ref{sec:othersampling}, and related questions about counting the number of graphs in a given graph space are covered in Section~\ref{sec:enumeration}. Section~\ref{sec:vignettes} employs the samplers from Section~\ref{sec:sampling}, examining the questions and decisions outlined in this introduction in the context of three separate applications of configuration model null models to study empirical network structure. Readers whose primary interest is understanding the practical consequences of configuration model choices are invited to skip Sections~\ref{sec:sampling}--\ref{sec:enumeration} and go directly to Section~\ref{sec:vignettes}, though the earlier sections establish the procedures employed therein.

\section{Markov chain Monte Carlo Sampling} 
\label{sec:sampling}

In this section we establish theoretical justifications for the use of Markov chain Monte Carlo (MCMC) methods to uniformly sample from graph spaces with a fixed degree sequence, with specific considerations for multiedges, self-loops, and vertex- or stub-labeling. In all methods presented in this section, a Markov chain over the desired space of graphs is designed to have a stationary distribution that is uniform over the entire space. We emphasize key differences between sampling stub-labeled and vertex-labeled graph spaces, and furnish pseudocode for all the MCMC sampling algorithms that we analyze.\footnote{Implementations in Python are available at \href{https://github.com/joelnish/double-edge-swap-mcmc}{https://github.com/joelnish/double-edge-swap-mcmc}}

We begin by reviewing the double edge swap Markov chain method for sampling stub-labeled loopy multigraphs, the easiest space for understanding the validity of the sampling procedure. We outline the three sufficient conditions (regularity, aperiodicity, connectivity) that combine to establish that random double edge swaps on stub-labeled loopy multigraphs have a unique and uniform stationary distribution. The corresponding lemmas and theorems are then reported, with references provided for known proofs, for stub-labeled simple graphs and stub-labeled multigraphs (without loops). 

Following the treatment of stub-labeled graph spaces, we then characterize Markov chains with stationary distributions that are uniform over vertex-labeled graph spaces. These chains have not previously been described, though they are closely related to existing methods for sampling the space of contingency tables with fixed marginals \cite{verbeek1985survey}, a problem from the statistics literature and discussed in the introduction. 

Sampling from spaces of loopy graphs (without multiedges) is not discussed in this section. Such spaces lack certain key properties necessary for sampling methods involving double edge swap routines to succeed. We elaborate on this matter in Section \ref{sec:othersampling}, where we also discuss other methods for graph sampling, including alternative Markov chains as well as direct sampling techniques.

\subsection{Edge swap Markov chains}

First developed for bipartite simple graphs \cite{besag1989generalized} and directed simple graphs \cite{rao1996markov}, Markov chain traversals of graph spaces are popular ways to sample from a variety of graph spaces \cite{miklos2004randomization,artzy2005generating,newman2003mixing}. If the Markov chain is constructed so that the stationary distribution of the chain is the uniform distribution over the desired graph space, samples taken from this chain at sufficiently spaced intervals (see the discussion of mixing times in Section \ref{sec:mixingtime}) can be treated as independent uniform samples from the space. 

The fundamental gadget underlying the approach is a randomized way of generating new graphs from existing graphs. Seemingly rediscovered multiple times \cite{hakimi1963realizability,ryser1957combinatorial,newman2003mixing,bienstock1994degree}, the most popular way to alter a graph without changing the degree sequence is the double edge swap, first suggested by Petersen in 1891 \cite{petersen1891theorie}, and depicted in Figure~\ref{doubleEdgeSwapFig}. Let $\{u_1,...,u_{k_u}\}$ denote the set of edge stubs for a vertex $u$ with degree $k_u$. Across the literature, double edge swaps are also sometimes referred to as degree-preserving rewirings \cite{cafieri2010loops,taylor1981constrained}, checkerboard swaps\footnote{Checkerboard swaps are frequently implemented by selecting 4 vertices at random \cite{artzy2005generating} while double edge swaps choose 2 edges at random. We focus on selecting edges at random as it is more efficient on sparse graphs. } \cite{stone1990checkerboard,gotelli1996null,artzy2005generating}, tetrads \cite{verhelst2008efficient}or alternating rectangles \cite{rao1996markov}. 

\begin{figure}
\centering
	\includegraphics[width=0.8\linewidth]{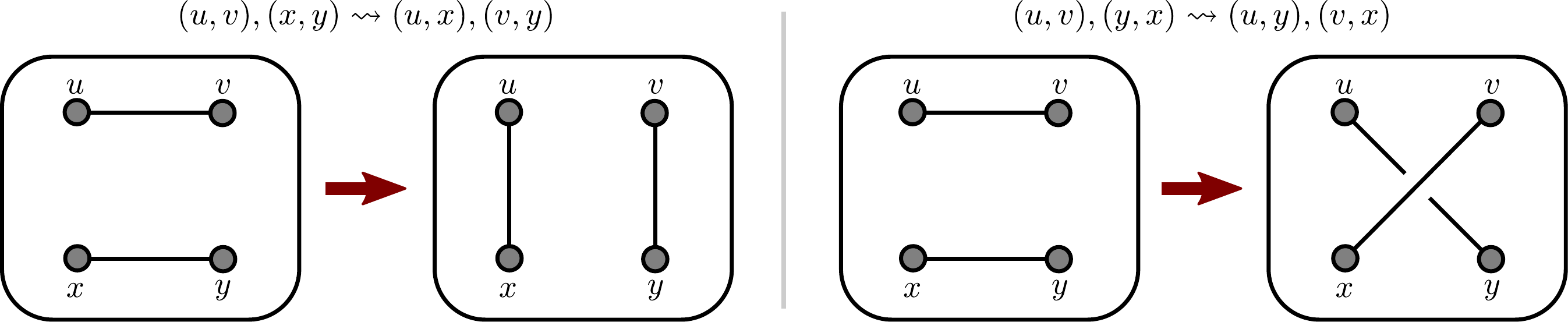}
	\vspace{-3mm}
	\caption{{\bf Double edge swaps.} Double edge swaps alter a graph's structure without changing the degree sequence. Each pair of edges may be swapped in two different ways: (left) $(u,v),(x,y)\leadsto (u,x),(v,y)$ and (right) $(u,v),(y,x) \leadsto (u,y),(v,x)$.}
	\label{doubleEdgeSwapFig}  
\end{figure}

\begin{definition}[Double Edge Swap, stub-labeled]
A \underline{stub-labeled double edge swap} replaces a pair of stub-labeled edges $(u_i,v_j)$ and $(x_p,y_q)$ with stub-labeled edges $(u_i,x_p)$ and $(v_j,y_q)$.
\end{definition}

Explicitly labeling stubs emphasizes that the stub-labeled double edge swap differs from its vertex-labeled version. That said, the notation of tracking stubs is largely unnecessary as the exact labels of stubs can be inferred in context and standard network analyses (of assortativity, modularity, etc.)~do not consider stub labels. For a pair of edges $(u,v)$ and $(x,y)$ there are two possible swaps, as shown in Figure~\ref{doubleEdgeSwapFig}. As a shorthand, we denote these possible swaps as $(u,v),(x,y)\leadsto (u,x),(v,y)$ and $(u,v),(y,x) \leadsto (u,y),(v,x)$.

In contrast to arbitrary edge rewires \cite{burda2003uncorrelated}, double edge swaps preserve the degree distribution of the graph. Notice, however, that some double edge swaps can create self-loops, e.g.~$(u,x),(u,y) \leadsto (u,u),(x,y)$, as well as multiedges, e.g.~when any produced edge replicates an existing edge. The way such swaps are handled has important consequences for the stationary distribution of the Markov chain.

Many of the properties of the double edge swap can be understood as graphical properties of the \textit{graph of graphs}, the state diagram of the Markov chain in the space of graphs. We construct the graph of graphs associated with a degree sequence by letting each graph with the specified degree sequence be a vertex and connecting two vertices (i.e.~graphs) with an edge if one double edge swap can transform one graph into the other. We use $\mathcal{G}(k)$ or $\mathcal{G}$ to generically denote a graph of graphs with a specified degree sequence $k= \{k_i\}_{i \in V}$. Throughout the text we only consider graph spaces with a given degree sequence, and as a consequence we almost always suppress the degree sequence $k$ from the notation, denoting a graph of graphs as simply $\mathcal{G}$.
With a few simple yet crucial modifications, sampling graphs using a random walk on $\mathcal{G}$ creates a Markov chain with a stationary distribution that is uniform over a desired graph space with a given degree sequence.

The statements in the following sections can be stated either in the language of Markov chains or in the language of graph properties of $\mathcal{G}$. To prove that samples from the Markov chain asymptotically obey a uniform distribution over a space of graphs, we show that by correctly specifying state transition probabilities, the chain satisfies three conditions:
\vspace{1mm}
\begin{itemize}
	\item[(i)] that the transition matrix of the chain is doubly stochastic ($\mathcal{G}$ is regular\footnote{A weighted directed graph is {\it regular} if every vertex has the same total out-degree weight and total in-degree weight.  For unweighted graphs, regularity implies all vertices have equal degree.}), 
	\item[(ii)] that the chain is irreducible (equivalently, $\mathcal{G}$ is strongly connected\footnote{A graph is {\it strongly connected} if every vertex can be reached from any other vertex. }),
	\item[(iii)] and that the chain is aperiodic ($\mathcal{G}$ is aperiodic\footnote{A graph is {\it aperiodic} if the greatest common divisor of the lengths of all cycles in the graph is one.}).
\end{itemize}
\vspace{1mm}
\noindent The regularity of $\mathcal{G}$ implies that the stationary distribution is uniform. A Markov chain that is both irreducible and aperiodic ($\mathcal{G}$ is connected and aperiodic) is said to be {\it ergodic}. This property guarantees that there is an unique stationary distribution that fully describes the long term behavior of the chain. Aperiodicity of $\mathcal{G}$ is often immediate and is particularly important if one wishes to subsample a Markov chain, a common strategy where only an infrequent set of samples (less sequentially correlated than the full set of samples) is retained. Once regularity and aperiodicity are established for loopy multigraphs, we show that with the appropriate modifications to transition probabilities, these properties also hold for the graph of graphs associated with any subspace of loopy multigraphs with a fixed degree sequence, whether vertex-labeled or stub-labeled. In contrast, connectivity of $\mathcal{G}$ (the irreducibility of the Markov chain) is not always guaranteed, and requires a non-trivial proof for many graph spaces, but is critical to ensuring that all possible graphs are sampled. 

\subsection{Markov chains on stub-labeled loopy multigraphs}

We begin by considering the simplest graph space for constructing and analyzing double edge swaps, $\mathcal{G}^{stub}_{l,m}$, where $stub$ denotes stub-labeled, $m$ denotes an allowance for multiedges, and $l$ denotes an allowance for loops. Further, let $M=\frac{1}{2} \sum_{i\in V} k_i$ denote the total number of edges in any graph in the graph space.

\begin{definition}[Graph of loopy multigraphs, stub-labeled]
 For some predefined degree sequence $k=\{k_i\}$, the \underline{graph of stub-labeled loopy multigraphs} $\mathcal{G}^{stub}_{l,m}=\{ \mathcal{V}^{stub}_{l,m},\mathcal{E}^{stub}_{l,m} \}$ is a directed graph, where the vertex set $\mathcal{V}^{stub}_{l,m}$ is the set of all stub-labeled loopy multigraphs with degree sequence $k$ and there is a directed edge $(G_1\to G_2)\in \mathcal{E}^{stub}_{l,m}$ iff there exists a stub-labeled double edge swap that transforms $G_1 \in \mathcal{V}^{stub}_{l,m}$ into $G_2 \in \mathcal{V}^{stub}_{l,m}$.
\end{definition}

For the space of loopy multigraphs, all edges in the graph of graphs $\mathcal{G}^{stub}_{l,m}$ are reciprocated: any double edge swap of distinct edges leads to a graph in the space and the double edge swap on $(u,v),(x,y) \leadsto (u,x),(v,y)$ can be undone by the ``reciprocal'' double edge swap $(u,x),(v,y) \leadsto (u,v),(x,y)$. 
Note however, that double edge swaps in other spaces are not necessarily reciprocated by the same number of swaps. 

We now show the three necessary conditions: that $\mathcal{G}^{stub}_{l,m}$ is regular, connected and aperiodic.

\begin{lemma} 
 \label{stub_multi_regular}
 $\mathcal{G}^{stub}_{l,m}$ is a regular graph. 
\end{lemma}
\begin{proof}
 For each graph $G_j \in \mathcal{V}^{stub}_{l,m}$ there are $\binom{M}{2}$ pairs of edges and $M(M-1)$ possible double edge swaps that each correspond to a unique graph-graph transition edge into and out of $G_j$. We immediately see that $\mathcal{G}^{stub}_{l,m}$ is $M(M-1)$ regular, where each vertex has $M(M-1)$ incoming and outgoing edges.
\end{proof}

Next, the following lemma, first proved by \cite{eggleton1979graph} and largely provided by Newman in \cite{newman2003mixing}, gives connectivity for stub-labeled loopy multigraphs with any specified degree sequence.
\begin{lemma} 
 \label{stub_multi_connected}
 $\mathcal{G}^{stub}_{l,m}$ is a strongly connected graph.
\end{lemma}
\begin{proof}
 First, we note that it is possible to permute stub labels using double edge swaps: for a graph $G_i \in \mathcal{V}^{stub}_{l,m}$ with vertex $u$ with degree at least $2$ (vertices with degree 1 have only a single possible stub labeling), a double edge swap  $(u_i,a_k),(b_\ell,u_j)\leadsto (u_i,b_\ell),(u_j,a_k)$ swaps two labeled stubs of $u$. Since double edge swaps allow for pairwise swaps of stubs, all possible stub-labelings within a given stub-isomorphism class of graphs are connected within $\mathcal{G}^{stub}_{l,m}$ (or any other stub-labeled space we discuss). The remainder of the proof therefore only requires showing that every stub-isomorphism class is connected to every other.
 
To complete the proof, we drop stub labels and will show how to construct a path from any $G_1=(V_1,E_1) \in \mathcal{G}^{stub}_{l,m}$ to any non-isomorphic $G_2=(V_2,E_2) \in \mathcal{G}^{stub}_{l,m}$ such that each step in the path creates and does not eliminate, edges in $E_2$. Let $\epsilon_{1,2} = |E^*_1 \setminus E^*_2|$, where the asterisks denote that the stub labels have been dropped from the edge sets. Since $\epsilon_{1,2}=0$ if and only if $G_1$ is isomorphic to $G_2$, it suffices to show that for any non-isomorphic graphs $G_1$ and $G_2$ there exists a neighbor of $G_1$, $G_3$, with $\epsilon_{3,2} \le \epsilon_{1,2} -1$. 

 Since $\epsilon_{1,2}>0$ there exists $(u,v)\in E^*_2\setminus E^*_1$. However, since the degrees of $u$ and $v$ are, respectively, the same in both $G_1$ and $G_2$, there must be edges $(u,x)$ and $(v,y)$ in $E^*_1\setminus E^*_2$. Performing the double edge swap $(u,x),(v,y)\leadsto(u,v),(x,y)$ creates a graph $G_3$ with edge $(u,v)$ and thus with $\epsilon_{3,2}\le \epsilon_{1,2}-1$. Since $\epsilon_{1,2}$ is finite, a repeated application of this argument eventually produces a path, and therefore $\mathcal{G}^{stub}_{l,m}$ is connected.
\end{proof}

\begin{lemma} 
 \label{stub_multi_aperiodic}
 $\mathcal{G}^{stub}_{l,m}$ is an aperiodic graph. 
\end{lemma}
\begin{proof}
 If $G\in \mathcal{V}^{stub}_{l,m}$ has only a single edge, $\mathcal{G}^{stub}_{l,m}$ is trivially aperiodic since $|\mathcal{V}^{stub}_{l,m}|=1$. If $G$ has two edges $(u,v)$ and $(x,y)$ then $\mathcal{G}^{stub}_{l,m}$ contains both a cycle of length 2 (because all transitions are reciprocated) and also a cycle of length 3: $(u,v),(x,y) \leadsto (u,x),(v,y)$ followed by $(u,x),(y,v) \leadsto (u,y),(x,v)$ and $(u,y),(v,x) \leadsto (u,v),(x,y)$. The greatest common divisor of the cycle lengths $2$ and $3$ is $1$, and therefore $\mathcal{G}^{stub}_{l,m}$ is aperiodic. 
\end{proof}

The following theorem assembles the above properties to establish the desired uniformity of the MCMC sampler.

\begin{theorem}
 \label{stub_MCMC} 
 A random walk on $\mathcal{G}^{stub}_{l,m}$ is ergodic and has a uniform stationary distribution.
\end{theorem}
\begin{proof}
 Since $\mathcal{G}^{stub}_{l,m}$ is strongly connected (Lemma \ref{stub_multi_connected}) and aperiodic (Lemma \ref{stub_multi_aperiodic}) random walks on $\mathcal{G}^{stub}_{l,m}$ are ergodic. Since $\mathcal{G}^{stub}_{l,m}$ is also regular (Lemma \ref{stub_multi_regular}) it has the unique stationary distribution $\vec{\frac{1}{|\mathcal{V}^{stub}_{l,m}|}}$.
\end{proof}

Thus, we conclude that a Markov chain defined as a random walk on $\mathcal{G}^{stub}_{l,m}$ in fact samples from the uniform distribution of stub-labeled loopy multigraphs, as desired. A similar MCMC approach can sample the other graph spaces under analysis here, though the proofs are slightly more involved. 

\subsection{Markov chains on other stub-labeled graph spaces}

We now show that with some care it is possible to construct Markov chains defined over the other stub-labeled graph spaces we have discussed such that their stationary distributions are also uniform. We establish this uniformity by deriving state transitions that ensure the chains are regular, connected, and aperiodic. Our results here apply to spaces of either simple graphs or multigraphs with a given degree sequence. The space of loopy graphs (without multiedges) with a given degree sequence is {\it not} connected by double edge swaps for all degree sequences and so we do not discuss it here; see Section~\ref{sec:othersampling} for more details on that space.

\begin{definition}[Graph of multigraphs and graph of simple graphs, stub-labeled]
 For a degree sequence $k = \{k_i\}$, the \underline{graph of stub-labeled simple graphs} $\mathcal{G}^{stub}_{s}=\{\mathcal{V}^{stub}_{s},\mathcal{E}^{stub}_{s} \}$ is a directed graph of simple graphs. For distinct $G_i$ and $G_j$ in $\mathcal{V}^{stub}_{s}$, a directed edge $(G_i \to G_j)$ is in $\mathcal{E}^{stub}_{s}$ if and only if there exists a double edge swap that transforms $G_i$ into $G_j$; for any double edge swap that would transform $G_i$ to a graph $G_j$ that is not in $\mathcal{V}^{stub}_{s}$, there instead exists a directed self-loop $G_i \to G_i$. The \underline{graph of stub-labeled multigraphs} $\mathcal{G}^{stub}_{m}$ is defined similarly for multigraphs, with subscripts of $m$ where appropriate.
\end{definition}

A critical difference between the definitions of $\mathcal{G}^{stub}_{s}$ and $\mathcal{G}^{stub}_{m}$ compared with the earlier definition of $\mathcal{G}^{stub}_{l,m}$ is the inclusion of directed self-loops $G_i \to G_i$ for each swap that would leave the space. This modification essentially employs the ``swap and hold'' \cite{artzy2005generating} (also called ``trial swap'' \cite{miklos2004randomization}) method to ensure the graph of graphs is regular.\footnote{In spaces featuring graphs without self-loops, each graph will have exactly $\sum_{i\in V} \binom{k_i}{2} $ swaps that could create self-loops; thus regularity is preserved if swaps that create self-loops either resample the current graph or are all ignored as possible swaps. There is a computational benefit from ignoring self-loop-creating edge swaps (as opposed to resampling the current graph), but it is likely small for most degree sequences.}

Indeed, we will now show that $\mathcal{G}^{stub}_{s}$ and $\mathcal{G}^{stub}_{m}$ are both regular and both aperiodic. As a result, extending Theorem \ref{stub_MCMC} only requires space-specific proofs of connectivity, which we provide.

\begin{lemma} 
 $\mathcal{G}^{stub}_{s}$ and $\mathcal{G}^{stub}_{m}$ are regular graphs. \label{stub_regular} 
\end{lemma}
\begin{proof}
 As in Lemma \ref{stub_multi_regular}, a graph $G_i$ in either space has $\binom{M}{2}$ pairs of edges, which correspond with $M(M-1)$ possible double edge swaps. Notice that any possible swap from $G_i$ to another graph $G_j$ in the space is reciprocated, while any swap that would go to a graph outside of the space corresponds with an incoming self-loop as constructed in the definition of $\mathcal{G}^{stub}_{s}$ and $\mathcal{G}^{stub}_{m}$. Thus, any graph $G_i$ in either of these two spaces has in-degree and out-degree $M(M-1)$.
\end{proof}

\begin{lemma}
    $\mathcal{G}^{stub}_{s}$ and $\mathcal{G}^{stub}_{m}$ are aperiodic graphs.
    \label{stub_aperiodic} 
\end{lemma}
\begin{proof}
    If there are any self-loops in the graph of graphs (where self-loops correspond to rejected swaps) and the graph of graphs is also connected then it is aperiodic. Meanwhile, if the graph of graphs does not have any rejected swaps (e.g.~when $\max_{i\in V} k_i<2$), then it has the exact same structure as $\mathcal{G}^{stub}_{l,m}$ and is thus aperiodic by Lemma \ref{stub_multi_aperiodic}.
\end{proof}

Before proving connectivity of the graph of graphs in the next lemma, we note that the proofs of Lemmas \ref{stub_regular} and \ref{stub_aperiodic} are easily and directly applied to any subspace of stub-labeled loopy multigraphs with fixed degree sequence (e.g., subspaces of graphs consisting of a single connected component, or subspaces with a constrained number of triangle motifs). However, despite the fact that regularity and aperiodicity are easy to establish for the graphs of graphs corresponding to such subspaces, proofs of their connectivity, if they are possible at all, require more complicated and subspace-specific constructions, and are considerably more involved. In fact, as noted above, for loopy graphs (without multiedges) connectivity does not hold for all degree sequences; see Section~\ref{sec:othersampling}. Below we establish the connectivity of $\mathcal{G}^{stub}_{m}$ and $\mathcal{G}^{stub}_s$ for any given degree sequence.

\begin{lemma}
 \label{stub_multi_noloops_connected}
 $\mathcal{G}^{stub}_{m}$ is a strongly connected graph.
\end{lemma}
\begin{proof}
 The proof that $\mathcal{G}^{stub}_{l,m}$ is connected (Lemma \ref{stub_multi_connected}) can be adjusted very slightly for the absence of self-loops. In the proof of Lemma \ref{stub_multi_connected}, if the two edges being considered for a double edge swap share an endpoint vertex then rewiring $(u,x)$ and $(v,x)$ creates the desired edge $(u,v)$ but also the self-loop $(x,x)$, and thus is not a valid swap as it would not stay within the space of loop-free multigraphs. But since $x$ has two edges contained in $E_1\setminus E_2$ and $x$ has the same degree in both the graph $G_2$ and $G_1$, there must exist at least one edge $(x,z)\in E_2\setminus E_1$, where $z \ne u$, $z \ne v$. Rewiring $(u,v)$ and $(x,z)$ in $G_2$ produces a neighboring graph $G_3$ with edge $(u,x)$ and thus $\epsilon_{1,3} \le \epsilon_{1,2}-1$.
\end{proof}

\begin{lemma}
 \label{stub_simple_connected} 
 $\mathcal{G}^{stub}_s$ is a strongly connected graph.
\end{lemma}

We do not provide a proof here as this result has been proven independently many times: in 1962 \cite{berge1962theory}, stated without proof in 1973 \cite{eggleton1973graphic}, proved twice in the same monograph but by different authors in 1981 \cite{eggleton1981simple,taylor1981constrained}, in 1994 \cite{bienstock1994degree}, and most recently in 2010 \cite{zhang2010traversability}.

\begin{theorem}
\label{stub_MCMC_all} 
A random walk on $\mathcal{G}^{stub}_{m}$ or $\mathcal{G}^{stub}_{s}$ is ergodic and has a uniform stationary distribution.
\end{theorem}
\begin{proof}
Being regular (by Lemma \ref{stub_regular}), connected (by Lemmas \ref{stub_multi_noloops_connected} and \ref{stub_simple_connected}), and aperiodic (by Lemma \ref{stub_aperiodic}) graphs, random walks on $\mathcal{G}^{stub}_{m}$ and $\mathcal{G}^{stub}_{s}$ are ergodic and have the unique stationary distributions $\vec{\frac{1}{|\mathcal{V}^{stub}_{m}|}}$ and $\vec{\frac{1}{|\mathcal{V}^{stub}_{s}|}}$ respectively.
\end{proof}

We conclude this subsection on sampling stub-labeled graph spaces with pseudocode for a uniform sampling algorithm. The important distinction between this algorithm and most incorrect algorithms (see Section~\ref{sec:directsample} for a further discussion of sampling algorithms known to be non-uniform) is that incorrect algorithms have a tendency to overlook the resampling step.\footnote{Implementations in Python are available at \href{https://github.com/joelnish/double-edge-swap-mcmc}{https://github.com/joelnish/double-edge-swap-mcmc}}

\subsection{Markov chains on vertex-labeled spaces}

For any analysis of simple graph null models, sampling from the vertex-labeled space is equivalent to sampling from the stub-labeled space: the two distributions are proportional within stub-isomorphism classes (see Section \ref{sec:enumeration} for details on this conversion). For non-simple graphs, the vertex-labeled and stub-labeled spaces are no longer cleanly proportional, but we show it is possible to adapt the double edge swap MCMC procedures to uniformly sample vertex-labeled graph spaces. We begin with the following definition, closely related to the double edge swap defined for stub-labeled spaces.

\begin{definition}[Double edge swap, vertex-labeled]
 A  \underline{vertex-labeled double edge swap} replaces pair of edges $(u,v)$ and $(x,y)$ with edges $(u,x)$ and $(v,y)$. 
\end{definition}

As in the stub-labeled setting, the vertex-labeled double edge swap leads to a Markov chain on the graph of vertex-labeled graphs, which we generically denote with $\mathcal{G}^{vert}$ (in contrast with $\mathcal{G}^{stub}$). In any graph space, stub-labeled double edge swaps map onto vertex-labeled double edge swaps simply by ignoring the stub-labeling: a vertex-labeled graph of graphs $\mathcal{G}^{vert}$ can be created by treating stub-isomorphic graphs within $\mathcal{G}^{stub}$ as a single graph in $\mathcal{G}^{vert}$. This construction of $\mathcal{G}^{vert}$ gives definitions for $\mathcal{G}^{vert}_{l,m}$, $\mathcal{G}^{vert}_{m}$, and $\mathcal{G}^{vert}_{s}$ as agglomerated, weighted, and directed, versions of the stub-labeled graphs of graphs $\mathcal{G}^{stub}_{l,m}$, $\mathcal{G}^{stub}_{m}$, and $\mathcal{G}^{stub}_{s}$, respectively. As a result, they immediately inherit the strong connectivity and aperiodicity properties of their respective stub-labeled spaces, as follows.

\begin{lemma}
 \label{vertex_conn} 
 $\mathcal{G}^{vert}_{l,m}$, $\mathcal{G}^{vert}_{m}$, and $\mathcal{G}^{vert}_{s}$ are strongly connected\footnote{Additionally, the graph space which allows multiedges and single self-loops is connected under edge swaps, while the graph space which allows only single edges, but potentially multi-self-loops, is disconnected unde edge swaps for some degree sequences\cite{nishimura2017swap}.}.
\end{lemma}
\begin{proof}
 Each of the vertex-labeled graph of graphs can be created by repeatedly combining vertices from the analogous stub-labeled graph of graphs until all stub-permutations of the same vertex-labeled graph have been combined together. Since iteratively combining vertices preserves connectivity, $\mathcal{G}^{vert}_{l,m}$, $\mathcal{G}^{vert}_{m}$, and $\mathcal{G}^{vert}_{s}$ inherit strong connectivity from $\mathcal{G}^{stub}_{l,m}$, $\mathcal{G}^{vert}_{m}$, and $\mathcal{G}^{stub}_{s}$.
\end{proof}

\begin{lemma}
 \label{vertex_aperiodic_s} 
 $\mathcal{G}^{vert}_{s}$, $\mathcal{G}^{vert}_{m}$, and $\mathcal{G}^{vert}_{l,m}$ are aperiodic graphs.
\end{lemma}
\begin{proof}
 For any fixed degree sequence, the proofs of Lemmas \ref{stub_multi_aperiodic} and \ref{stub_aperiodic} either apply directly, and thereby establish aperiodicity, or the proofs of Lemmas \ref{stub_multi_aperiodic} and \ref{stub_aperiodic} do not apply because they necessitated double edge swaps between two graphs in the same stub-isomorphism class. However, even in this case, the double edge swap between graphs in the same stub-isomorphism class implies there is a self-loop in the graph of graphs, and the graph of graphs is thus aperiodic.
\end{proof}

\begin{algorithm}[H]
 \caption{stub-labeled MCMC \label{alg_stub} }
\begin{algorithmic}
\REQUIRE {initial graph $G_0$, graph space (simple, multigraph, or loopy multigraph)}
\ENSURE {sequence of graphs $G_i$}
 \FOR {$i<$ number of graphs to sample} 
 \STATE choose two edges at random
 \STATE randomly choose one of the two possible swaps
 \IF {edge swap would leave graph space} 
\STATE resample current graph: $G_{i} \leftarrow G_{i-1}$
\ELSE
\STATE swap the chosen edges, producing $G_{i}$
\ENDIF
\ENDFOR
\end{algorithmic}
\end{algorithm}

\begin{figure}
\centering
	\includegraphics[width=0.9\linewidth]{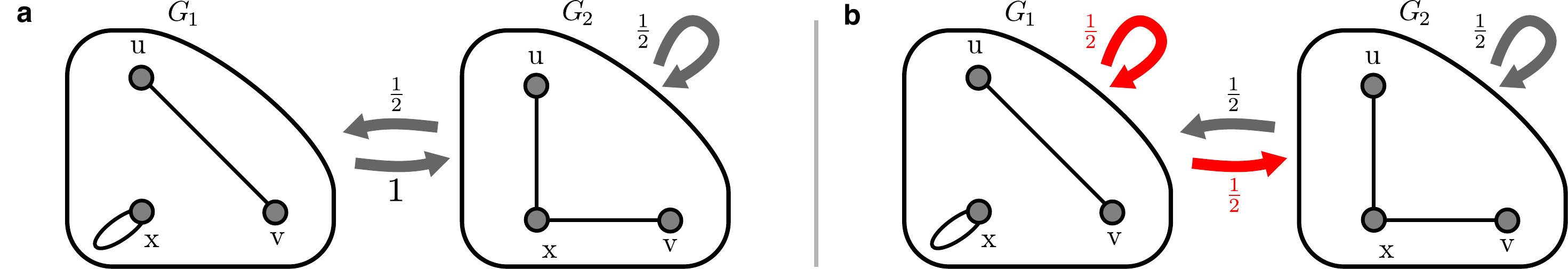} 
	\vspace{-3mm}
	\caption{{\bf Transition probabilities for uniform sampling.} The graph of vertex-labeled loopy multigraphs $\mathcal{G}^{vert}_{l,m}( \{2,1,1 \})$ contains two possible graphs $G_1$ and $G_2$. (a) A random walk on this graph of graphs has $Pr(G_1\to G_2) = 1$ but $Pr(G_2\to G_1)=\frac{1}{2}$ and therefore its corresponding Markov chain will not have a uniform stationary distribution since the graph of graphs is not regular. (b) If transition probabilities are modified such that each graph has equal in-degree weight and out-degree weight (i.e.~transition mass), and these weights are equal to each other, the corresponding Markov chain will have a uniform stationary distribution and will therefore sample each graph with equal probability. }
	\label{vertex_example}
\end{figure}

While connectivity and aperiodicity of vertex-labeled graphs of graphs follow directly from the properties of the stub-labeled spaces, regularity is more complicated. The analysis of stub-labeled graphs of graphs relied on the fact that each swap had a unique reciprocal swap. This reciprocity is not present in vertex-labeled graphs of graphs. For example, consider $\mathcal{G}^{vert}_{l,m}$ on a degree sequence as simple as $\{2,1,1 \}$. As shown in Figure \ref{vertex_example}(a), the graph of graphs $\mathcal{G}^{vert}_{l,m}( \{2,1,1 \})$ contains only two possible graphs: $G_1$ (with self-loop $(x,x)$ and edge $(u,v)$) and $G_2$ (with two adjacent edges $(u,x)$ and $(v,x)$). Every swap originating in $G_1$ creates $G_2$ (both swaps of $(x,x)$ and $(u,v)$ create $(u,x)$ and $(v,x)$), but only one of the two possible swaps originating in $G_2$ reaches $G_1$ ($(u,x),(v,x)\leadsto(u,v),(x,x)$ corresponds to $G_2 \to G_1 $ while $(u,x),(x,v)\leadsto(u,x),(x,v)$ corresponds to $G_2 \to G_2$).
If unaltered, a random walk on $\mathcal{G}^{vert}_{l,m}( \{2,1,1 \})$ has the non-uniform stationary distribution $(Pr(G_1) = \frac{1}{3}, Pr(G_2) =\frac{2}{3})$. 
Restoring the regularity of $\mathcal{G}^{vert}_{l,m}( \{2,1,1 \})$, as in Figure \ref{vertex_example}(b), is achieved by rejecting the swap $G_1\to G_2$ with probability $\frac{1}{2}$ and instead looping back to $G_1$. Figure \ref{vertex_example} shows a difficulty arising from self-loops; vertex-labeled swaps of multiedges suffer a similar problem with a similar resolution.
As we will show, an extra layer of rejection sampling suffices to restore the uniform stationary distribution for any vertex-labeled graph. 

There are two natural ways to implement rejection sampling for vertex-labeled graphs, which we provide in Algorithm~\ref{alg_vertex_basic} and in the supplemental material, Algorithm~\ref{alg_vertex_complex}. The simpler of the two approaches, Algorithm~\ref{alg_vertex_basic}, employs a rejection sampling that modifies all swaps $G_i \to G_j$, $i\ne j$, to have probability $\frac{1}{M(M-1)}$. The following lemma demonstrates that Algorithm \ref{alg_vertex_basic} achieves this uniform probability on all possible swaps.

\begin{lemma}
 \label{vertex_regular_s} 
 A Markov chain defined by a random walk on $\mathcal{G}^{vert}_{l,m}$, $\mathcal{G}^{vert}_{m}$, or $\mathcal{G}^{vert}_{s}$ with transition probabilities given by Algorithm \ref{alg_vertex_basic} has a doubly stochastic transition matrix.
\end{lemma}
\begin{proof}
Algorithm \ref{alg_vertex_basic} randomly selects two edges $e_1$ and $e_2$ and also selects one of the two possible ways to swap $e_1$ and $e_2$. The goal is to make all swaps equally probable. If $e_1$ or $e_2$ is a self-loop then the potential swap is rejected with probability $\frac{1}{2}$. If not rejected, then if both edges connect the same vertices (i.e.~$e_1 = e_2$), the swap is made with probability $\frac{2}{w_{e_1} (w_{e_1}-1) }$, where $w_{e_1}$ is the multiplicity of edge $e_1$, and otherwise the swap is made with probability $\frac{1}{w_{e_1}w_{e_2}}$. If no swap is made or the proposed swap would not change the graph (e.g.~$(u,v)(v,v)\leadsto(u,v)(v,v)$) the current graph is resampled by the chain. To see that these rejection probabilities give all swaps an equal overall probability of success, consider the following table of double edge swaps cases, which presents the form of each possible swap, the number of such possible swaps, and the acceptance probabilities used by Algorithm \ref{alg_vertex_basic}. 

\vspace{2mm}
 \small
 \begin{center} 
 \begin{tabular}{c|c |c |c|c}
 \hline
 $e_1,e_2 \leadsto e_3,e_4$ & \# possible
 & \multicolumn{3}{c}{$Pr($perform swap$)$}  \\
& 
 stub-labeled swaps & if $e_1$ or $e_2$ is a self-loop & if $e_1=e_2$ & if $e_1\not=e_2$ \\
 \hline
 $(u,v),(x,y)\leadsto(u,x),(v,y)$ & $w_{uv} w_{xy}$ & -- & -- &$1/(w_{uv} w_{xy})$\\
 $(u,x),(x,v)\leadsto(u,v),(x,x)$ & $w_{ux} w_{xv}$ & -- & -- &$1/(w_{ux} w_{xv})$\\
 $(x,x),(u,v)\leadsto(u,x),(x,v)$ & $2w_{xx} w_{uv}$ & $1/2$ & -- &$1/(w_{xx} w_{uv})$\\
 $(u,u),(x,x)\leadsto(u,x),(u,x)$ & $2w_{uu} w_{xx}$ & $1/2$ & -- &$1/(w_{uu} w_{xx})$\\
 $(u,x),(u,x)\leadsto(u,u),(x,x)$ & $ \binom{w_{ux}}{2}$ & -- & $1/\binom{w_{ux}}{2}$ & --\\
 \hline
 \end{tabular}
 \end{center}
\normalsize
\vspace{2mm}

On a pair of edges containing a self-loop, both swaps result in the same edges post-swap, giving a factor of $2$ to the number of possible swaps of that type. Notice also that multiplying the factors in a given row results in the same overall transition mass, 1, for each row. Thus, every swap is equally likely with probability $ \frac{1}{M(M-1)}$ and the transition matrix is doubly stochastic.
\end{proof}

As a direct result of Lemma \ref{vertex_regular_s}, the sum of edge weights directed to any graph in $\mathcal{G}$ with these transition probabilities equals one. Algorithm \ref{alg_vertex_basic} can be understood as changing general double edge swap stub-labeled spaces into double edge swap vertex-labeled spaces for any subspace of loopy multigraphs with a fixed degree sequence. Assembling Lemmas \ref{vertex_conn}, \ref{vertex_aperiodic_s} and \ref{vertex_regular_s} gives the following theorem.

\begin{figure}
\centering
	\includegraphics[width=0.85\linewidth]{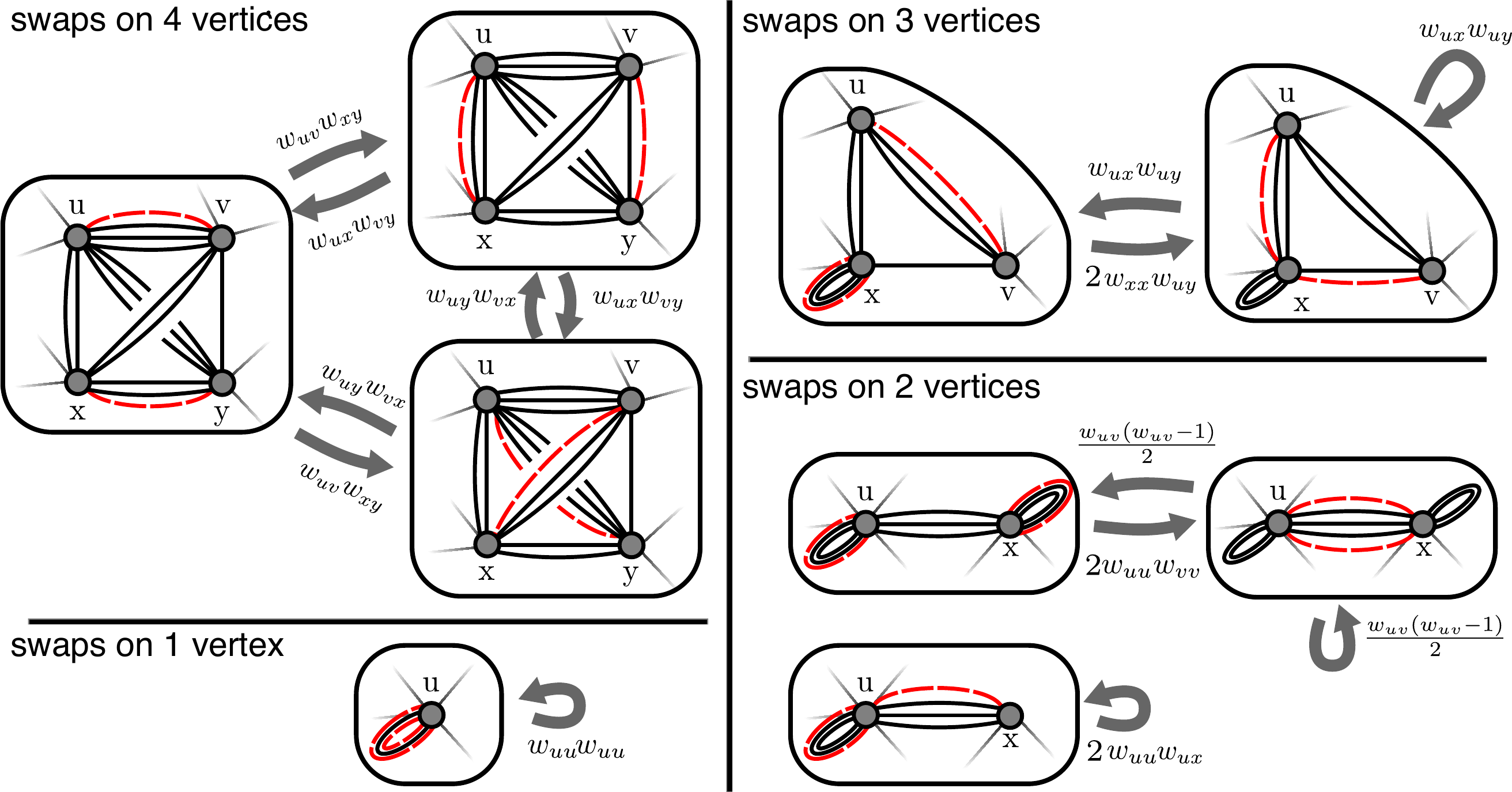}
	\vspace{-3mm}
	\caption{{\bf Possible double edge swaps for Algorithms \ref{alg_vertex_basic} and \ref{alg_vertex_complex}.}  Algorithms \ref{alg_vertex_basic} and \ref{alg_vertex_complex} can be understood by examining the probabilities of each double edge swap that is possible for a given graph. The diagram's labeled transitions give the number of possible double edge swaps that transition between each graph, organized by the number of unique vertices involved, where $w_{ij}$ are the edge multiplicities in the originating graph. Uniform sampling of a graph space can be achieved by down-sampling transitions to be equal in both directions.}
	\label{EdgeSwapCases}
\end{figure}

\begin{theorem}
\label{vertex_MCMC_all} 
A Markov chain on $\mathcal{G}^{vert}_{l,m}$, $\mathcal{G}^{vert}_{m}$, or $\mathcal{G}^{vert}_{s}$ with transition probabilities given by Algorithm \ref{alg_vertex_basic} is ergodic and has a uniform stationary distribution.
\end{theorem}
\begin{proof}
Lemma \ref{vertex_regular_s} gives that $\vec{\frac{1}{|\mathcal{V}_{l,m}|} }$, $\vec{\frac{1}{|\mathcal{V}_{m}|} }$, and $\vec{\frac{1}{|\mathcal{V}_{s}|} }$ are the respective stationary distributions; strong connectivity (Lemma \ref{vertex_conn}) and aperiodicity (Lemma \ref{vertex_aperiodic_s}) give that the Markov chain is ergodic.
\end{proof}

We conclude this subsection on sampling vertex-labeled graph spaces with pseudocode for the uniform sampling algorithm, Algorithm \ref{alg_vertex_basic}, used in the above proofs. A more efficient but more complicated approach is given in Algorithm \ref{alg_vertex_complex} in the supplemental material. This more efficient algorithm achieves regularity by computing both the forward and reverse probabilities of any given double edge swap according to the cases in Figure \ref{EdgeSwapCases}. It then down-samples (rejects) the higher probability swap to have the same probability as the lower probability swap. For example, in Algorithm~\ref{alg_vertex_complex} a double edge swap of the edges $(u,v)$ and $(x,y)$ (on distinct vertices $u,v,x,y$) to form $(u,y)$ and $(x,v)$ is accepted with probability $\mathrm{min}(1, \frac{ w_{uy} w_{xv} }{ w_{uv}w_{xy} })$, whereas Algorithm~\ref{alg_vertex_basic} accepts this swap with probability $\frac{1}{w_{uv}w_{xy}}$. While Algorithm~\ref{alg_vertex_complex} requires calculating these forward and reverse probabilities for each swap, we observe empirically that it mixes substantially faster on degree sequences with higher degrees. 

\begin{algorithm}[t]
\caption{vertex-labeled MCMC}\label{alg_vertex_basic}
\begin{algorithmic}
\REQUIRE initial graph $G_0$, graph space (simple, multigraph, or loopy multigraph)
\ENSURE sequence of graphs $G_i$
\FOR {$i<$ number of graphs to sample}
\STATE {choose two distinct edges $e_1$ and $e_2$ uniformly at random}
\STATE {randomly choose one of the two possible swaps}
\IF{edge swap would leave the graph space }
    \STATE resample current graph: $G_{i} \leftarrow G_{i-1}$
\ELSE
    \STATE {$P\leftarrow 1$}
\IF {$e_1$ and $e_2$ are copies of the same multi-edge}
    \STATE {$P \leftarrow \frac{2P}{w_{e_1} (w_{e_1}-1) }$}
\ELSE
    \STATE {$P \leftarrow \frac{P}{w_{e_1}w_{e_2}}$}
\ENDIF
\IF { $e_1$ or $e_2$ are a self-loop}
    \STATE {$P\leftarrow \frac{1}{2}P$}
\ENDIF
\IF {$Unif(0,1)<P$}
    \STATE {swap chosen edges to produce $G_{i}$}
\ELSE
    \STATE {resample current graph: $G_{i} \leftarrow G_{i-1}$}
\ENDIF
\ENDIF
\ENDFOR
\end{algorithmic}
\end{algorithm}

\subsection{Mixing times}
\label{sec:mixingtime}

As discussed in the previous section, a MCMC sampler based on double edge swaps will eventually sample from $\mathcal{G}^{stub}_{l,m}$, $\mathcal{G}^{stub}_{m}$, $\mathcal{G}^{stub}_{s}$, $\mathcal{G}^{vert}_{l,m}$, $\mathcal{G}^{vert}_m$ and $\mathcal{G}^{vert}_s$ uniformly. A natural question, and one of practical importance, is how many swaps it takes before a sample from the Markov chain is negligibly correlated with the starting graph. This question is usually studied in the language of {\it mixing time}, the number of steps in a Markov chain required to produce a sample a prescribed distance from the stationary distribution of the chain \cite{levin2009markov}. A Markov chain on a graph space is said to be {\it rapidly mixing} if the mixing time can be expressed as a polynomial in the number of vertices. Empirical investigations tend to support the notion that the mixing times of edge swap MCMC samplers tend to be reasonable and not prohibitive \cite{milo2003uniform,newman2003mixing}. 
Theoretical investigations have identified various conditions on the degree sequence $k$ which rigorously support these observations \cite{cooper2007sampling,greenhill2015switch}. However, the case of general $k$ is yet to be fully understood.

As first demonstrated in \cite{sinclair1992improved}, the most common argument to derive mixing time bounds uses a multicommodity flow argument, and the most common focus has been on regular simple graphs and regular directed graphs. Thus far, rapid mixing has been proved for double edge swap MCMC methods on simple graphs with regular degree sequences \cite{cooper2007sampling}, regular directed graphs \cite{greenhill2011polynomial}, and half-regular and almost half-regular bipartite graphs \cite{miklos2013towards,erdHos2015approximate}. Beyond regular graphs, there are bounds based on the minimum and maximum degrees, which give polynomial mixing in time $O(k_{max}^{14}M^9(M\log(M) - \log{\epsilon}))$ if $3\le k_{max} \le \frac{1}{4}\sqrt{M} $ \cite{greenhill2015switch}. A loosely related set of investigations shows that while the shortest paths in $\mathcal{G}^{vert}_s$ can be approximated to within a factor of $7/4$, finding the shortest path is NP-hard \cite{bienstock1994degree,erdos2013swap}.

Mixing time results for non-simple graphs are, by comparison, poorly developed. While stub- and vertex-labeled spaces have different transition probabilities and different structures, recall that vertex-labeled graphs of graphs can be created by repeatedly merging vertices in the corresponding stub-labeled graph of graphs. As a result, the total diameter of a vertex-labeled graph of graphs $\mathcal{G}^{vert}$ is necessarily always smaller than the corresponding stub-labeled graph of graphs $\mathcal{G}^{stub}$, but the additional layer of rejection sampling in vertex-labeled MCMC chains may lead mixing times to be large for degree sequences where multiedges and self-loops are more common. Determining the conditions, if any exist, in which the smaller diameter of vertex-labeled spaces corresponds to faster mixing times is an interesting open question. 

In practice, there are well-accepted diagnostics to numerically assess the quality of MCMC mixing \cite{gelman2014bayesian}. One popular method is to compare the variance inside a sequence to variance across multiple sequences, while other methods analyze the correlation inside a sequence. These diagnostics are typically performed on a sequence of graph statistics, rather than directly on a sequence of graphs. One complicating factor of using inter-sequence variation to assess convergence is the difficulty in finding independent starting graphs with which to start the chain \cite{Brooks98someissues}. Ultimately, when considering the potential effect of mixing times, it is important to gauge the risk of a slow mixing time (and thus a biased sampler), against errors associated with uniformly sampling from an inappropriate space, as is often the case with stub-matching.

\section{Other sampling methods and other null models}
\label{sec:othersampling}

Edge swap Markov chains are not the only means of sampling from configuration models, nor are configuration models the most appropriate random graph null model for all analyses. In this section we briefly review other techniques for sampling configuration models, as well as other random graph null models that have been usefully employed in other contexts. Very little is known about the adaptation of the methods in this section to vertex-labeled graph spaces, but such adaptations are discussed when known. 

\subsection{Direct sampling and other sampling methods}
\label{sec:directsample}

Edge swap Markov chain methods work by randomly manipulating an initial graph to produce a new graph, with the idea being that the stationary distribution of this random process is designed to be uniform over the graph space. In contrast, ``direct'' methods sample the same space by constructing one graph at a time without any dependence on previous samples. Sampling uniformly from graph spaces is closely related to enumerating the number of graphs in a given space, a task commonly known as graph enumeration \cite{bayati2010sequential} (see Section \ref{sec:enumeration} for more on these connections).

The stub-matching procedure pioneered by Bollob\'as \cite{bollobas1980probabilistic}, also called the pairing model and discussed in Section~\ref{sec:stub}, is an example of a direct method for sampling the space of loopy multigraphs with a given degree sequence. Stub matching begins with a prescribed number of half edges or stubs attached to each vertex in an otherwise empty graph and then randomly joins pairs of unmatched stubs to form a graph. The graph created by this procedure is a uniform sample from the space of stub-labeled loopy multigraphs. 

For more restricted graph spaces, i.e.~those that omit self-loops and/or multiedges, stub matching must be adapted. Early work on directly sampling simple graphs with specified degree seqeuences focused on regular graphs \cite{mckay1990uniform}, with later results giving approximately uniform sampling for more general degree sequences \cite{bayati2010sequential}. The simplest adaptation of stub matching for restricted graph spaces, e.g.~for simple graphs, is to use {\it rejection sampling}: complete a stub-matching procedure, and if the resulting graph is not in the graph space, reject the sample. This process is repeated until a simple an admissible graph is returned. Using rejection sampling, an unrejected graph is a proper uniform sample from the graph space. Unfortunately, rejection sampling for simple graphs can take exponential time---exponential in the size of the graph---for some degree sequences with degrees that increase in the size of the graph. In contrast to rejection sampling, a more efficient approach is to apply {\it sequential importance sampling} \cite{blitzstein2011sequential}, where edges are possibly rejected during the construction process (rather than waiting until the end to reject the output graph). The basic idea behind sequential importance sampling is to guide the matching process by rejecting edges that push the stub-matching process toward overrepresented simple graphs. Interestingly, a sequential importance sampling technique whereby each edge is rejected with a probability $\frac{k_i k_j}{4M}$ is sufficient to approximately sample uniformly for graph spaces where the max degree $k_{\text{max}}$ obeys $k_{\text{max}} = O(M^{1/4-\tau})$ for some $\tau>0$ \cite{bayati2010sequential}, but this asymptotic statement does not furnish any clear guarantees for an empirical graph of a fixed size. 

Other modifications to stub matching exist, usually posed in the context of creating simple graphs, and each with a mix of desirable and undesirable properties. One approach freely matches stubs, which may create a self-loop or multiedge, but such an edge is immediately removed via a double edge swap \cite{leskovec2016snap}. In contrast to rejection or importance sampling, this loop and multiedge rewiring approach ensures that a graph from the desired space is produced by each full run of the algorithm, which may dramatically improve the rate at which samples are produced. However, it unfortunately biases the sampling in ways that are not yet described or understood. Other methods knowingly generate biased simple graphs via constrained stub-matching, and each sample's relative probability is calculated in order to perform {\it a posteriori} bias corrections that reweight the samples to guarantee uniformity \cite{del2010efficient}. Again, there do not yet exist bounds on the convergence of such methods to the uniform distribution desired. More exotic direct sampling procedures include the so-called {\it Go with the Winners} algorithm \cite{aldous1994go} applied to graph generation \cite{milo2003uniform}. This method employs stub-matching on a collection of graphs in parallel, replacing failed attempts to create simple graphs with cloned copies of non-failed attempts, eventually producing a set of admissible graphs. Finally, it is possible to define an alternative Markov chain based on {\it perfect matchings} to uniformly sample regular simple graphs \cite{jerrum1990fast}; this method can be adapted to non-regular degree sequences but without efficiency guarantees.

Constructive procedures for determining whether a given degree sequence is {\it graphical} (that there exists a simple graph with the given degree sequence \cite{gallai1960graphs}), notably the Havel-Hakimi algorithm \cite{havel1955remark,hakimi1962realizability}, are highly non-uniform direct sampling procedures. The Havel-Hakimi algorithm is useful as a starting point for MCMC methods in contexts where one starts with a degree sequence but no corresponding simple graph---Havel-Hakimi is guaranteed to efficiently produce a simple graph, which one can then use as the initial state of a MCMC method.

\subsection{Markov chains for sampling other spaces}
\label{sec:othermarkov}

Markov chains other than ``double edge swap'' chains can be used to traverse other graph spaces with specified degree sequences, notably spaces of {\it connected graphs}, spaces of {\it loopy graphs} (without multiedges), and spaces of {\it directed graphs}.

{\it Loopy graphs (without multiedges).}
Sampling methods based on the double edge swap Markov chain discussed in Section \ref{sec:sampling} are unfortunately not sufficient for sampling uniformly from the space of loopy graphs with a specified degree sequence. The main challenge to sampling is that for certain degree sequences the double edge swap Markov chain does not connect the entire space of loopy graphs. For example, the degree sequence $\{2,2,2\}$ in the space of loopy graphs admits both a triangle graph and a graph consisting of $3$ self-loops, but on both graphs it is easy to see that any proposed double edge swaps would create a multiedge. Thus the two graphs in the space are not connected by any sequence of double edge swaps that remain in the space of loopy graphs, and this lack of connectivity applies to both the stub- and vertex-labeled spaces. Generalizing this observation, it is the case that the space of loopy graphs is connected for any degree sequence that can wire a simple graph and is neither the degree sequence of a path, $\{2,2,...,2\}$, nor that of a clique, $\{n-1,n-1,...,n-1 \}$ \cite{nishimura2016sampling}. Alternatively, if the Markov chain is modified to occasionally employ a three-edge {\it triangle-loop swap} (the swap $(u,u),(v,v),(w,w)\leadsto(u,v),(v,w),(w,u)$ and the reciprocal swap $(u,v),(v,w),(w,u)\leadsto(u,u),(v,v),(w,w)$), a basic modification of Algorithm \ref{alg_stub} and Algorithm \ref{alg_vertex_basic} suffices to sample uniformly from these spaces; see \cite{nishimura2016sampling} for more details.

{\it Connected graphs.} Many real-world graphs are connected, either by design (e.g.~the architecture of the Internet \cite{maslov2004detection}) or by virtue of how they were measured (using snowball sampling \cite{goodman1961snowball} or other traversal techniques). It is known that double edge swaps can rewire any connected graph to any other with the same degree sequence \cite{taylor1981constrained,bienstock1994degree}. Therefore, if one correctly rejects swaps that would leave the space of connected graphs then Theorems \ref{stub_MCMC_all} and \ref{vertex_MCMC_all} would apply. Thus, we can conclude that there exists a double edge swap MCMC sampler of connected graphs whose stationary distribution is the uniform distribution over connected graphs with a prescribed degree sequence. However, there is no computationally expedient way to certify connectivity\footnote{Checking connectivity can be done in $O(\sqrt{|V|})$ time with each change to the graph \cite{eppstein1997sparsification}.} of the resulting graph for a proposed swap. A useful heuristic solution is to only check connectivity after completing a longer sequence of swaps \cite{gkantsidis2003markov,viger2005efficient}. 

\begin{figure}
	\centering
	\includegraphics[width=0.95\linewidth]{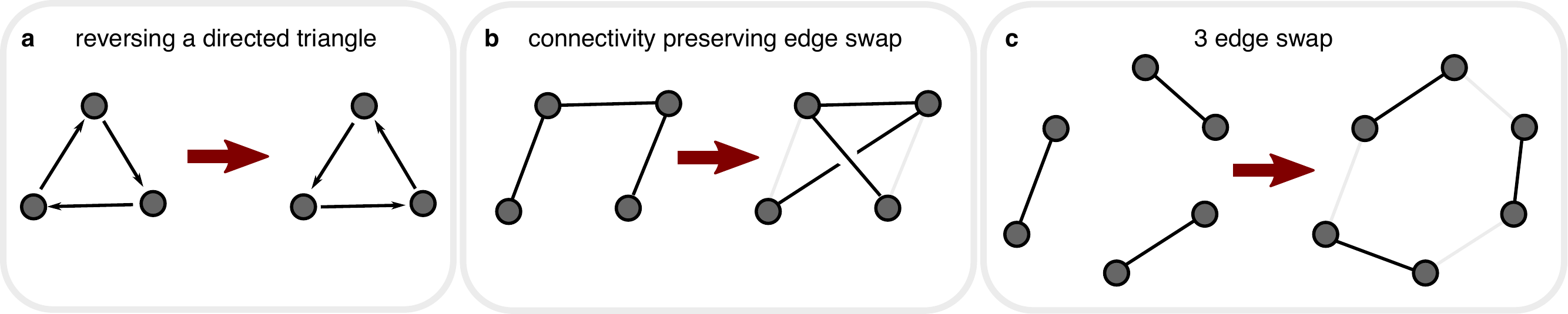}
	\vspace{-3mm}
	\caption{{\bf Other varieties of edge swaps.} In addition to the double-edge swap, a number of other edge swaps have been used to construct MCMC samplers of specific spaces. These auxiliary swap routines are necessary to ensure the underlying graph of graphs associated with the chain is connected.  (a) Markov chains on directed graph spaces require an additional swap that reverses a directed triangle. (b) The $k$-Flipper swap for $k=3$, swapping the endpoints of length-3 paths, preserves connectivity spaces of connected graphs. (c) Swaps involving more than $2$ edges enable sampling graph spaces with more complicated constraints.}
	\label{fig:otherSwaps}
\end{figure}

A more expedient approach for sampling connected simple graphs with a given degree sequence follows from a Markov chain defined by a different swap: a {\it $k$-Flipper} Markov chain in a given graph space selects length-$k$ paths uniformly at random (typically employed with $k=3$, see Figure \ref{fig:otherSwaps}(b)) and swaps the endpoints of the path \cite{mahlmann2005peer}. This swap clearly results in a graph that has the same connectivity before and after the swap. What is less clear is that a chain utilizing this swap does not necessarily explore the full space of connected graphs with a specified degree sequence; a chain occasionally utilizing a small additional swap (dubbed the {\it bowtie swap}) is required to ensure that the graph of graphs is connected, and thus samples the entire space of connected graphs \cite{feder2006local}. This chain has a uniform stationary distribution, and some mixing time results are known under mild assumptions \cite{feder2006local}. Of note, $k$-Flipper techniques cannot be extended (in any obvious way) to graph spaces that allow self-loops as a $k$-Flipper swap is unable to ever create a self-loop. Studies of the space of connected graphs have been focused on simple graphs, and it is an open question to understand what role the choice of stub-labeling vs.~vertex-labeling has in studies of connected multigraphs.

{\it Directed graphs.} Sampling directed graphs using edge swap Markov chains introduces new subtleties that are not present when sampling undirected graphs. 
Most importantly, a directed graph has two separate degree sequences, the in-degree sequence and out-degree sequence, and one may wish to fix either or both of these. The two sequences are coupled because the sum of the graph's in-degree must equal the sum of its out-degrees. Furthermore, in order for a graph of directed graphs to be connected under edge swaps, a {\it directed triangle reversal} swap is needed, Figure \ref{fig:otherSwaps}(a), which reverses the direction of a three-edge cycle \cite{kleitman1973algorithms,rao1996markov,lamar2011directed}. Sampling both stub-labeled and vertex-labeled directed graphs builds on a similar theoretical framework as undirected graphs \cite{carstens2015Proof,carstens2016curveball,carstens2016switching}.

Broadly speaking, as richer network models are considered the sensible value of uniform distributions as statistical null models decreases. Developing an appropriate null model for richer networks, where directed graphs are one example, requires carefully considering and modeling a hypothesized generative processes. For example, a directed version of a citation network should roughly obey causality constraints (cycles would indicate past papers citing future papers), and the statistical properties of such a network might be best captured by comparing it to the output of generative model that explicitly accounts for publication date.

\subsection{Distributions over graphs with edge weights}

In applications, graphs often have scalar weights associated with their edges. In some special cases, these weights are integers and can be interpreted as the number of edges between vertices. The graph then is, in fact, a multigraph and the techniques discussed thus far may be applied directly. However in all other cases, where the weights do not have a natural edge multiplicity interpretation, specifying a null model becomes substantially more difficult. In particular, a decision must be made regarding whether the null model should preserve just vertex degrees, or both vertex degrees and vertex total weight (the sum of the edge weights associated with a vertex). Even in the former, simpler case, a null model that preserves vertex degrees must choose carefully how to additionally randomize the edge weights.

To see the difficultly of this problem, consider any double-edge swap process where at least two edge weights are distinct. The original weights could be assigned at random to the pair of rewired edges, corresponding to a null model in which edge existence and edge weight are entirely independent, but this would not preserve the total weight associated with the involved vertices. On the other hand, edge existence and edge weights could be chosen to be coupled in some way, but that requires actively placing assumptions on the nature of the relationships.  In general, devising a procedure that preserves vertices' degrees and their total weights while randomizing the edges and weights is an open problem.

\subsection{Other distributions over graph spaces} 

Lastly, we note other varieties of distributions over graph spaces that are sometimes employed as null models. Most of these models depart from configuration models in that the constraint to an exact degree sequence $k=\{k_i\}_{i\in V}$ is relaxed. Often these models exhibit specified well-studied degree sequences in expectation.

The random graph model most closely related to configuration models is the {\it Chung-Lu model} \cite{chung2002average}. Rather than being specified by a fixed degree sequence, the Chung-Lu model is parametrized by a sequence of expected degrees, and for most well-behaved degree sequences the model correctly samples graphs with these expected degrees. 

In the context of producing simple graphs, one can also generate a graph via stub-matching and then remove all self-loops and/or multiedges that have been generated, a procedure called the {\it erased configuration model} or {\it Molloy-Reed configuration model} \cite{molloy1995critical,newman2001random}. Deleting an edge necessarily changes the degree sequence, and thus this technique will not sample only graphs with the specified degree sequence. For sufficiently bounded degree sequences, it has been shown that asymptotically there will be only $O(1)$ such deletions in large graphs \cite{molloy1995critical}. Thus, when the degree sequence lacks large degrees and for applications robust to a small number of edge deletions, the erased configuration model may provide a suitable approximation to the uniform distribution over simple graphs.

A separate and significant literature on random graph null models studies ensembles of graphs that are the result of random growth processes. The {\it Price model} \cite{de1965networks,price1976general}, also known as the preferential attachment model \cite{barabasi1999emergence}, generates random graphs with heavy-tailed degree sequences (though many other generative processes also generate such degree sequences \cite{mitzenmacher2004brief,clauset2009power}). Graphs generated by the Price model have very different structural properties than graphs generated by configuration models with the same expected degree sequence: asymptotically almost surely, graphs generated by the Price model are somewhere dense, while for the corresponding degree sequences, graphs generated by the erased configuration model (or Chung-Lu model) are nowhere dense (and in fact have bounded expansion, a stronger property) with high probability \cite{demaine2014structural}. In other words, graphs that are common under one model are extremely rare under the other, and vice versa. Other network growth processes include uniform growth \cite{callaway2001randomly}, again resulting in graphs with properties different from graphs grown under the Price model. For empirical graphs that may have resulted from a growth process, comparing the properties of the graph to the properties of an ensemble of random graphs generated from a growth model may be appropriate. 

Many null models other than configuration models are sampled using Markov chains. For example, Markov chains can be constructed to sample graphs with fixed degree-degree correlations, specifically by specifying each sampled graph to have a fixed joint degree-degree matrix\footnote{A {\it degree-degree matrix} is a matrix $C$ where entry $C_{k,k'}$ denotes the number of edges between vertices of degree $k$ and vertices of degree $k'$. A graph with a given degree-degree matrix also has a fixed degree sequence, which can be easily reconstructed from the degree-degree matrix \cite{stanton2012constructing}.} \cite{amanatidis2008graphic,stanton2012constructing,czabarka2015realizations}; direct sampling methods exist for this space as well \cite{bassler2015exact}. There is a non-trivial relationship between graphs with fixed degree-degree matrices and connected graphs: connectivity imposes constraints on a degree-degree matrix, e.g.~a connected graph of more than three vertices cannot have any degree-one--degree-one connections. Swaps that involve more edges (e.g.~Figure \ref{fig:otherSwaps}(c)) have been tailored to attempt to satisfy more complex constraints such as a fixed number of triangles or fixed component sizes \cite{tabourier2011generating}. However, even many-edge swaps may fail to connect the space for some constraints. For instance, the space of graphs with a fixed number of triangles is disconnected, even for triple or quadruple edge swaps \cite{nishimura2017swap}. Another approach is to allow graphs which break some constraints, but bias edge swaps towards satisfying constraints, such as those used to sample graphs which satisfy constraints on the counts of arbitrary subgraphs of fixed size \cite{orsini2015quantifying}.

{\it Exponential random graph models} (ERGMs) furnish non-uniform distributions over graph spaces that increase the relative probability of observing certain structural properties, and are typically sampled using Markov chain methods \cite{snijders2002markov}, though the mixing times of these chains are sometimes known to be very poor \cite{bhamidi2011mixing,chatterjee2013estimating}. ERGMs generally focus on simple graphs, though some recent work has extended ERGMs to multigraphs \cite{desmarais2012statistical,krivitsky2012exponential,chandrasekhar2014tractable}; identifying differences between ERGMs specified in vertex-labeled vs.~stub-labeled spaces is an open question. A different non-uniform {\it triadic closure Markov Chain}, related to the {\it Strauss model} (a specific ERGM) \cite{strauss1986general}, has also been proposed and studied for its abilities to replicate empirical subgraph frequencies in social networks \cite{ugander2013subgraph}.

Lastly, there is an enormous literature on models of community structure in networks. The most prominent such model is the {\it stochastic block model} \cite{holland1983stochastic}, which generalized the affiliation model \cite{frank1982cluster}. The stochastic block model has also been adapted to model overlapping (mixed-membership) community structure \cite{airoldi2008mixed}, community structure in bipartite networks \cite{larremore2014efficiently}, and hierarchical community structure \cite{peixoto2014hierarchical}. Other related graph null models include the degree-corrected stochastic block model \cite{karrer2011stochastic} and the block two-level Erd\H{o}s-R\'enyi (BTER) model \cite{kolda2014scalable}. The degree-corrected stochastic block model merges the stochastic block model with techniques from the Chung-Lu model to target an expected degree sequence.

\section{Graph enumeration}
\label{sec:enumeration}

Graph enumeration---counting the number of graphs within a space---relates directly to the uniform sampling problems discussed in this paper. Given a vertex-labeled graph $G$, we can calculate the number of stub-labeled graphs that are isomorphic to $G$, highlighting the difference in size and composition between stub- and vertex-labeled spaces, as shown, for example, in Figure~\ref{fig1}.

By efficiently enumerating this correspondence, it is possible to use a simple reweighting scheme to convert a uniform sample taken from one graph space to a uniform sample under another graph space. While theoretically sound, this approach can fail dramatically in practice for many graph spaces. Graphs that are frequent in one distribution can be enormously different from the graphs that are frequent under the other distribution, meaning that unreasonably large sample sizes are required to overcome biases; see Section~\ref{sec:geometers} for an illustration of this with an empirical degree sequence from a collaboration network. 

{\it Labeled graph spaces.}
The correspondence between vertex-labeled and stub-labeled graph enumerations is straightforward. For a vertex-labeled graph $G=(V,E)$ with a degree sequence $k=\{k_i\}$, we define $q_{\text{simple}}(G)$ as the number of stub-labeled simple graphs that correspond to a vertex-labeled simple graph $G$. The set of $k_i$ stubs for vertex $i$ can be arranged in $k_i !$ unique permutations, and this simple counting argument applied to the entire vertex set shows that:
\begin{align}
	q_{\text{simple}}(G)&=\prod_{i=1}^n k_i !.
\end{align}
This count depends only on the degree sequence $\{ k_i \}$ and not any other property of $G$. In other words, for a fixed degree sequence we see that each graph $G$ in the vertex-labeled space has the same number of stub-labeled graphs that correspond to it. Notice that this is true of the two simple graphs examples in Figure~\ref{fig1}(d,e). As a result, for simple graphs---and only for simple graphs---the relative sizes of the the isomorphism classes are the same in the vertex-labeled and stub-labeled spaces. Thus, an ensemble of random vertex-labeled simple graphs can be converted into a sample of stub-labeled simple graphs by randomly assigning stub-labels to each graph in the ensemble, and an ensemble of  stub-labeled simple graphs can be regarded as a sample of vertex-labeled simple graphs by simply ignoring stub labels.  

For graphs with multiedges or self-loops, it is still possible to count the number of stub-labeled graphs that correspond to each vertex-labeled graph, but now the multiplicity depends on more than just the degree sequence. The quantities are derived by adjusting $q_{\text{simple}}(G)$, the numerator in each quantity, for the number of identical configurations involving multiedges and/or self-loops.
Let $w_{ij}$ be the integer number of edges between vertices $i$ and $j$. For a single self-loop $w_{ii}=1$, again counting the number of edges. The multiplicities for each space are then as follows:
\begin{align}
	q_{\text{loopy}}(G) &= 
	q_{\text{simple}}(G) \times
	\frac{1}{\prod_{i=1}^n w_{ii}! (2^{w_{ii}})}\\
	q_{\text{multi}}(G) &= 
	q_{\text{simple}}(G) \times
	\frac{1}{\prod_{i<j}w_{ij}!}\\
	q_{\text{loopy multi}}(G) &= 
	q_{\text{simple}}(G) \times 
	\frac{1}{\prod_{i=1}^n w_{ii}! (2^{w_{ii}})} \times \frac{1}{\prod_{i<j}w_{ij}!}. 
\end{align}

The conversion factors in the equations above can be enormous, illustrating that, as stated above, the graphs that are prevalent in one distribution can be extremely different from those that are prevalent in the other distribution. As a result, a conversion between stub-labeled and vertex-labeled spaces is an infeasible approach to sampling from the less easily sampled space.

{\it Unlabeled graph spaces. }
For any enumeration related to the space of unlabeled graphs (isomorphism classes, see Figure~\ref{fig1}(c), efficient counting is unfortunately infeasible. Let $p_{\text{simple}}(G)$ be the number of vertex-labeled simple graphs that correspond to an unlabeled graph $G$. It is well known that $p_{\text{simple}}(G) = n!/|\text{Aut}(G)|$, where $|\text{Aut}(G)|$ is the size of the automorphism group of $G$, i.e.~the number of distinguishable vertex graph labelings. Determining $|\text{Aut}(G)|$ is polynomial-time equivalent to determining if two vertex-labeled graphs in the group are isomorphic \cite{mathon1979note}, making it as computationally difficult as the famous graph isomorphism problem \cite{babai2016graph}, for which the best known algorithm is quasipolynomial. Enumerating the size of the isomorphism class for loopy graphs, multigraphs, and loopy multigraphs is at least as hard. Thus, there are no known practical and efficient means of transferring between unlabeled and labeled graph spaces. 

This reasoning also tells us that any sampling method that could produce a uniform sample from the space of unlabeled graphs $G$ with a specified degree sequence would furnish a way to count $|\text{Aut}(G)|$, and thus must take at least quasipolynomial time (unless graph isomorphism is in the complexity class P). It is therefore unlikely that the uniform distribution over unlabeled graphs will see a polynomial time direct sampler, or a Markov chain sampler with a polynomial mixing time.

\section{Applications}
\label{sec:vignettes}

In this section we use three real-world examples that demonstrate how a configuration model can be used as a null model, employing the sampling procedures outlined in Section \ref{sec:sampling}, and how the choice of graph space can have substantial impact on hypothesis tests and scientific conclusions. The first example studies a graph of collaborations among researchers to show that the choice of null model graph space greatly impacts null distributions of degree correlations, leading to varying conclusions about the meaning of the observed degree correlation in the network. The second example studies a graph of interactions among barn swallows to show that the choice between vertex-labeled and stub-labeled spaces is non-trivial and directly impacts conclusions about the underlying animal behavior. Finally, the third example uses a graph of social support in South Indian villages to demonstrate that the vertex clusters found by modularity maximization, a popular community detection method traditionally based on the stub-labeled loopy multigraph configuration model, are sensitive to the choice of underlying graph space. Together, these examples illustrate the practical differences between graph spaces and show how the methods presented in this paper can be applied\footnote{Results in this section utilize code available at \href{https://github.com/joelnish/double-edge-swap-mcmc}{https://github.com/joelnish/double-edge-swap-mcmc} with convergence assessed via trace plots, autocorrelation, and effective sample size analyses.}.

\subsection{Degree assortativity in a collaboration network}
\label{sec:geometers}

Degree assortativity measures the extent to which pairs of connected vertices tend to have similar degrees. This degree-degree correlation is an easily computable and single-valued summary of edge patterns in a graph, and it has been used to shed light on the organizational differences between broad categories of social, biological, and technological networks \cite{newman2002assortative,newman2003mixing}. It is most commonly computed as the Pearson correlation between degrees of vertices that are connected in the network.  It is defined as
\begin{equation}
 r = \frac{\frac{1}{M}\sum_{(u,v) \in E} k_u k_v - \mu_k^2}{\sigma^2_{k}},
 	\label{r_deg}
\end{equation}
where $\mu_k$ and $\sigma^2_k$ are the mean and variance of the vertex degrees across stubs in the network. Positive degree-degree correlations ($r>0$) are commonly interpreted as degree assortativity, while negative correlations ($r<0$) are interpreted as degree disassortativity, but meaningful interpretations of $r$ require that we first quantify the possibility that degree-degree correlations are solely a consequence of the specific degree sequence (see, for example, {\it structural disassortativity} described in \cite{boguna2004cut}). In this application of configuration models, we show that not only does the choice of graph space dramatically shift the null distribution of  degree-degree correlations, but that it can even affect the sign of the expected value of the correlation and effectively invert the conclusions drawn from hypothesis tests.

Degree assortativity is common in social networks, and collaboration networks are commonly thought to be no exception, due to collaborations between extremely productive researchers. Here we consider a collaboration network among computational geometry researchers, where vertices represent researchers and edges represent co-authorship on a paper or book.  The data come from the Computational Geometry Database \cite{jones2002computational} and consist of 9,072 vertices and 22,577 edges. In a collaboration network a $c$-author publication induces a $c$-clique in the graph, because every pair of the $c$ co-authors will share an edge, $c(c-1)/2$ edges in total from a $c$-author publication. A collaboration network is naturally a multigraph since researchers often collaborate on multiple papers together, but there are no self-loops by construction. 

\begin{figure}[t]
	\centering
	\includegraphics[width=0.85\linewidth]{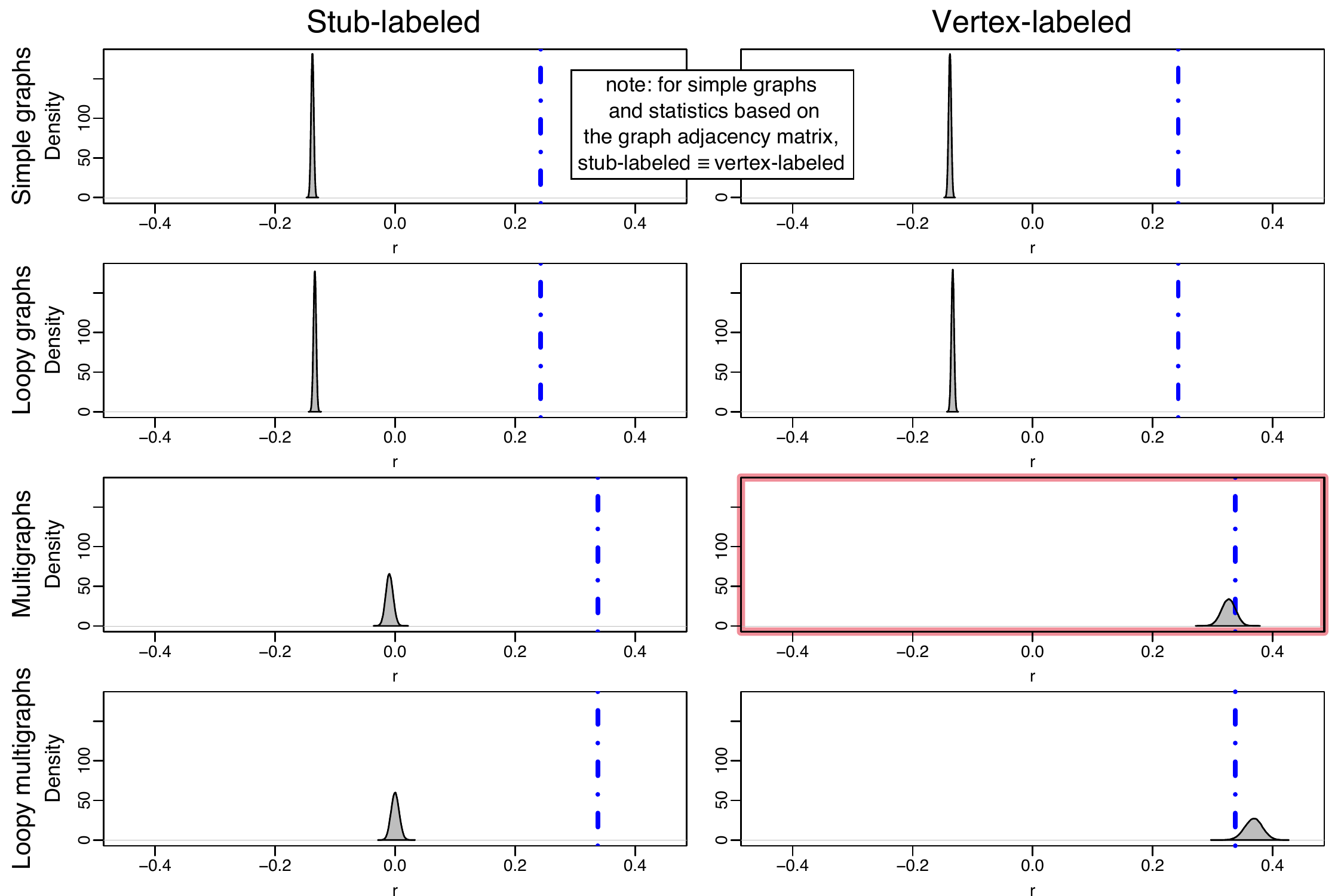}
	\caption{{\bf Degree assortativity of the geometers collaboration network.} Distributions of  degree assortativity corresponding to configuration models over various graph spaces are shown in grey, and the degree assortativity of the geometers collaboration graph is shown in blue. The thick red border around the vertex-labeled multigraph panel indicates the space chosen by answering the three guiding questions listed in Section~\ref{subsec:choosing}. For the spaces of multigraphs and loopy multigraphs, the configuration models uses the degree sequence of the multigraph collaboration network, while for the spaces of simple graphs and loopy graphs, the configuration models uses the degree sequences of the simplified collaboration network. Due to the fact that degree assortativity is a function of the graph adjacency matrix, distributions of assortativity over simple graphs (top row) are identical for both stub- and vertex-labeled spaces.}
	\label{fig:geoAssort}
\end{figure}

In Section \ref{subsec:choosing} and Figure \ref{fig2}, we listed and considered three questions to guide the choice of graph space, which we now answer in order. First, due to the construction of the collaboration graph, the network does not allow self-loops. Second, geometers can co-author multiple papers, so the network allows for multiedges. Third, the crossing of two edges in the multigraph is nonsensical---it is meaningless for author A's first collaboration with author B to be matched with author B's second collaboration with author A, and vice-versa---and therefore this collaboration network should be considered to be a vertex-labeled multigraph.

Although the collaboration network is a multigraph, a researcher might consider ``simplifying'' the observed network into the space of simple graphs by removing all duplicate edges between pairs of vertices, or equivalently thresholding all edge multiplicities at one. Although not applicable here, if the observed graph were to contain self-loops, an analogous removal of self-loops would be necessary to ``simplify'' the graph. Networks are sometimes simplified for convenience, stemming from a desire to analyze a binary simple graph using familiar tools. Simplification may also have a scientific basis, if, for example, the question of interest did not concern the number of relations between a pair of vertices but only whether or not any relation existed. Regardless, we demonstrate here that the decision to simplify can greatly impact conclusions.

Figure \ref{fig:geoAssort} shows distributions of degree assortatitivity over the different configuration models described in this paper, where the correlations of the empirical graphs (the original and the simplified) are shown as blue dashed lines, and the null distributions based on correctly sampled configuration models (Section~\ref{sec:sampling}) are shown as grey probability densities.\footnote{Although the space of loopy graphs is not necessarily connected under double-edge swaps, it can be shown to be connected for the degree sequence of this collaboration network, allowing the use of Algorithm 2 or 3 from Section~\ref{sec:sampling}; see \cite{nishimura2016sampling} for details.} Note immediately that many of the null distributions in Figure~\ref{fig:geoAssort} have almost no overlap in their distributional mass, illustrating two key practical implications of null model selection. First, comparing the panels in each column illustrates a direct impact of the inclusion or exclusion of self-loops and/or multiedges. Second, comparing the panels in each row indicates that, although the space of vertex-labeled graphs is nested within the space of stub-labeled graphs, the frequency at which each vertex-labeled graph appears in the stub-labeled space is so dramatically non-uniform that the ranges of  degree-degree correlations under each null distribution appear disjoint.

Most importantly, the null distribution differences shown in Figure~\ref{fig:geoAssort} lead to conflicting study conclusions. All four stub-labeled configuration models---which our decision framework identify as incorrect models---suggest that the observed collaboration graph is far more assortative than a random graph with the same degree sequence. However, this conclusion is dramatically tempered when using the vertex-labeled multigraph configuration model that was identified by answering the three questions of Section~\ref{subsec:choosing}. Furthermore, if one incorrectly allowed self-loops and sampled the space of vertex-labeled loopy multigraphs one might erroneously conclude that the collaboration network was slightly {\it dis}assortative. The dramatic variation of degree-degree correlations among null models, shown in Figure~\ref{fig:geoAssort}, highlights the importance of correctly choosing a graph space, and avoiding the default null model of stub-labeled loopy multigraphs associated with straightforward stub matching.

\subsection{Trait assortativity in a barn swallow interaction network}\label{sec:swallows}

Trait assortativity measures the extent to which pairs of connected vertices tend to have similar scalar-valued traits. This pairwise correlation is calculated using the same formula as degree assortativity in Eq.~\eqref{r_deg}, but with degrees replaced with trait values \cite{newman2002assortative,newman2003mixing}. As with degree assortativity, measurements of trait assortativity provide clues as to how particular traits are related to the arrangement of a network's edges. And again, as with degree assortativity, large or small values of trait assortativity are uninterpretable without first understanding the distribution of values which might be observed by random chance. In this application of configuration models, we show once more that scientific conclusions are highly sensitive to the graph space chosen as a null model, applying the methods of this paper to a multigraph of interactions among barn swallows and a trait that quantifies the birds' plumage color.

Past studies have shown that plumage color of the Colorado barn swallow (\textit{Hirundo rustica erythrogaster}) is associated with reproductive success \cite{safran2005dynamic}, but it is unknown if this is due to genetic incompatibility between birds of different colors or if it is due to the preferential mixing of birds by color. To investigate whether there is evidence that swallows preferentially interact with other swallows of similar color, we consider network and trait data describing a population of $17$ Colorado barn swallows collected during the 2014 breeding season \cite{levin2016stress}. Each vertex in the network represents a swallow, and each edge represents an interaction: an interaction was recorded between bird pairs whenever their proximity tags registered a close encounter, with interactions aggregated over 15 hours and measured across three days \cite{levin2015performance,levin2016stress}. Researchers also recorded the color of each bird's ventral plumage as a scalar, standardizing colors between bird sexes. To determine whether birds of similar color interacted more than one would expect by chance, while controlling for the fact that some birds have higher interaction counts than others, we compare the observed assortativity by color to the distribution of assortativity values for networks with identical degrees (i.e., interactions counts) but with their interactions randomized.

We now apply the three questions of Section \ref{subsec:choosing} and Figure \ref{fig2} to guide the choice of null model graph space. First, due to the fact that a bird cannot interact with itself, the network does not allow for self-loops. Second, because pairs of swallows may interact multiple times during the data collection period, the network allows for multiedges. Third, the crossing of two edges in the multigraph is nonsensical due to their temporal ordering---it is meaningless for bird A's first interaction to be paired with bird B's second interaction, and vice-versa---and therefore this interaction network is a vertex-labeled multigraph.

As in the previous application, a slight change in the scientific question could change the graph space selected by the three questions. Specifically, if the researchers wished to determine whether birds of similar color tended to {\it ever} interact with each other, the network should be ``simplified'' by reducing all multiedges to single edges, creating a vertex-labeled simple graph in which an edge is present between any pair of swallows that interacted at any point during data collection. It is tempting to think that this simplification will not be impactful---after all, only 34\% of interacting bird pairs interacted more than once, and only 11\% interacted more than twice. However, we now show that this is not the case.

\begin{figure}[t]
\centering
	\includegraphics[width=0.85\linewidth]{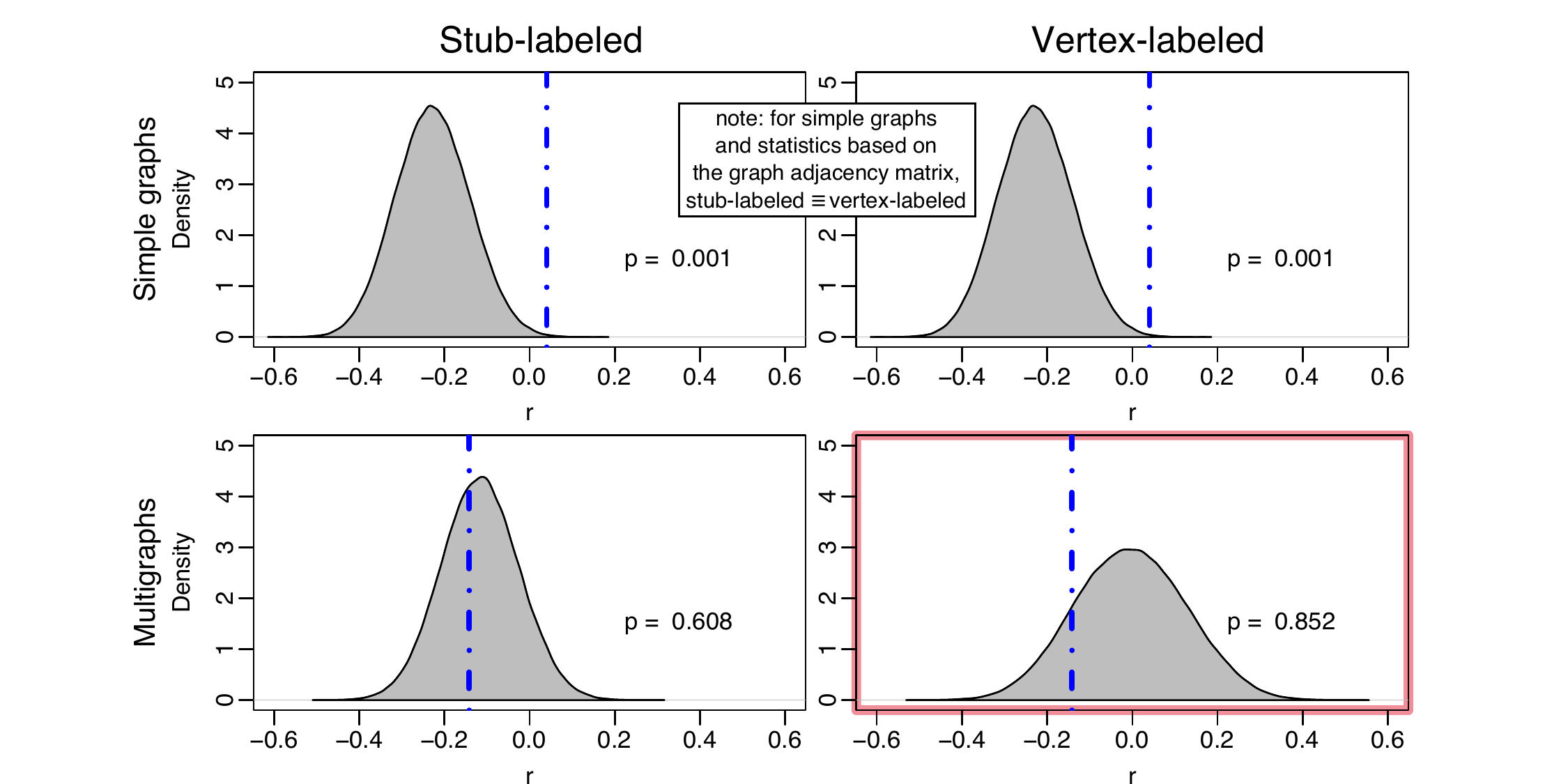}
	\caption{{\bf Color assortativity of the barn swallow interaction network.} Null distributions of color assortativity associated with the uniform distribution over simple graphs, stub-labeled multigraphs, and vertex-labeled multigraphs. The vertical blue line denotes the observed assortativity for the simplified (top row) original (bottom row) barn swallow networks. The thick red border around the vertex-labeled multigraph panel indicates the space chosen by answering the three guiding questions listed in Section~\ref{subsec:choosing}. Also depicted are the (upper-tailed) $p$-values, i.e.~the proportion of the null distribution assortativity values that are greater than the respective observed assortativity. Due to the fact that trait assortativity is a function of the graph adjacency matrix, distributions of assortativity over simple graphs (top row) are identical for both stub- and vertex-labeled spaces.}
	\label{swallows}
\end{figure}

Figure~\ref{swallows} shows color assortativity distributions for the simple graph configuration model and for vertex-labeled and stub-labeled multigraph configuration models, as well as $p$-values for the corresponding one-sided hypothesis tests testing for positive color assortativity (i.e.~whether swallows of similar color preferentially interact). Thus, in each case the $p$-value is equal to the proportion of the null distribution graphs with assortativity values greater than the observed value; small $p$-values are indicative that the observed assortativity is noteworthy. Once more, the choice of configuration model has a substantial and significant impact on the null distributions of color assortativity.  An analysis based on simple graphs would conclude that the presence or absence of interaction is significantly correlated with plumage color ($p=0.001$). However, the related analysis based on vertex-labeled multigraphs---the analysis identified by the three questions from Section~\ref{subsec:choosing}---concludes that there is no evidence that the number of interactions is significantly correlated with plumage color ($p=0.852$). 

This application reveals another, more subtle aspect of choosing a graph space. Due to the fact that both degree and trait assortativity are computed as a Pearson correlation, it is often assumed that in the absence of correlations, i.e., when edges are placed at random, $r=0$ \cite{newman2002assortative}, and that $r>0$ and $r<0$ indicate assortative and disassortative mixing patterns, respectively. However, as shown in Figures~\ref{fig:geoAssort}-\ref{swallows}, zero is the incorrect point for comparison; the distribution of color assortativity in Figure~\ref{swallows} is centered around zero for only one of the three graph spaces shown. Moreover, the simplified network has a near-zero assortativity, indicated by the blue dashed line, yet when compared with its null distribution from a simple configuration model, it is clear that the interaction presence/absence (simple) network is significantly assortative by plumage color. Thus, the choice of configuration model affects not only the scientific hypothesis being tested and its conclusion, but also the baseline against which we should anchor our intuition for correlations in networks. 

\subsection{Community detection in a South Indian village social support network}
\label{sec:modularity}
Community detection is a fundamental task of network science in which the vertices are divided into groups (also called clusters or communities) based solely on the patterns of edges. Often, communities are defined as groups of vertices that are more densely connected to each other than one would expect by chance. Community detection provides a course-grained summary of the network which enables further study of its large-scale organization and may also reveal correlations between vertex attributes and global network structure. Partitions of vertices produced by community detection have been used in a wide variety of applications, including studies of large-scale online social network structure \cite{leskovec2009community}, evolutionary constraints of malaria parasites \cite{larremore2013network}, and constructing experimental treatment groups for randomized controlled trials on networks \cite{ugander2013graph}. 

There are many approaches to community detection in networks \cite{fortunato2010community}, with one of the most popular being modularity maximization \cite{newman2004finding}. Modularity measures the strength of community structure in a network for a particular division of the vertices into groups, and its maximization is based on the premise that communities are groups of vertices that are more densely linked to each other than one would expect by chance---that is, than one would expect, were the edges of the network arranged randomly.  More precisely, modularity  is the average difference between the observed network adjacency matrix $A$ and its expectation $\mathbb E[A | k]$, under a configuration model null model, across all within-group edges in the network. In particular, modularity assumes a stub-labeled loopy multigraph configuration model, for which the expected number of edges between any two vertices $i$ and $j$, with degrees $k_i$ and $k_j$, respectively, would be $\mathbb{E}_{l,m}^{stub}[A_{ij} |  k] = k_i k_j / 2M$.\footnote{Expectations over the Chung-Lu model \cite{chung2002average} and expectations over the stub-labeled loopy multigraph configuration model are identical under a mild assumption on the skew of the degree distribution, that $\max_{i,j} k_i k_j / \sum_\ell k_\ell \le 1$. Thus, for stub-labeled loopy multigraphs, either model may be used to produce the estimate $k_i k_j / 2M$, but as we shall see, this is not the case for other graph spaces, for which the Chung-Lu model cannot be used.} 
The widely used modularity $Q$ is therefore defined as
\begin{equation}
	Q = \frac{1}{2M}\sum_{i,j} \left (A_{ij} - \frac{k_i k_j}{2M} \right ) \delta(g_i,g_j)\ ,
	\label{q_newman}
\end{equation}
where $A$ is the network adjacency matrix, $g_i$ is the community assignment of vertex $i$, and $\delta$ is the Kronecker delta which restricts the sum to within-group edges. 

The null model of modularity maximization, as it is written above, is the space of stub-labeled loopy multigraphs, yet this space is not necessarily an appropriate null model for many real-world networks. Modularity is often used to analyze simple graphs, and this can lead to unexpected or undesirable community partitions \cite{massen2005identifying,cafieri2010loops}. If a simple graph is sufficiently large, sufficiently sparse, and its degree sequence is sufficiently bounded, then the expected number of edges between two vertices in the space of simple graphs is asymptotically the same as the expectation in the space of stub-labeled loopy multigraphs, i.e.~$\mathbb E_{s}[A_{ij}|  k ] \approx k_i k_j / 2M$ \cite{newman2010networks} (where $s$ denotes simple graphs). Thus, Eq.~\eqref{q_newman} will produce asymptotically correct values for simple graphs. However, for finite simple graphs, we lack guarantees about the accuracy of Eq.~\eqref{q_newman}. The definitions and methods introduced in this paper now enable us to estimate these expectations to arbitrary accuracy by first identifying the correct graph space (Section~\ref{subsec:choosing}) and then sampling from it appropriately (Section~\ref{sec:sampling}). We now show that the choice of configuration model, and in particular the choice of a vertex-labeled model, meaningfully changes the results of community detection. 

For this investigation we analyzed a network of social support relationships in a pair of South Indian villages collected by Power \cite{power2015religious}. The number of edges between two members of the villages corresponds to the number of different social supports between them. Due to the differential meaning of each support, for a pair who share $m$ mutual supports, there are $m$ possible ways these can be shared, not $m!$. Thus, the dataset indicates that it belongs to the space of vertex-labeled multigraphs by answering the questions of Section~\ref{subsec:choosing}: self-loops are nonsensical, multiedges exist, and vertices are labeled but stubs are not.  

In order to redefine modularity for an arbitrary graph null model, we rewrite the expected number of edges between two vertices of degree $k$ and $k'$ as $\mathbb{E} [C_{k,k'}] / n_{k}n_{k'}$, where $n_k$ is the number of vertices in the network with degree $k$ and $\mathbb{E} [C_{k,k'}]$ is the expected number of edges between all vertices of degrees $k$ and $k'$ respectively under the specified null model.  We then rewrite modularity in generic form, based on $\mathbb E[C_{k,k'}]$,
\begin{equation}
	Q_{\text{generic}} = \frac{1}{2M}\sum_{i,j} \left (A_{ij} - \frac{\mathbb E[C_{k_i,k_j}]}{n_{k_i} n_{k_j}} \right ) \delta(g_i,g_j)\ .
	\label{q_ckk}
\end{equation}
To change the null model, we need only change the distribution of graphs over which $\mathbb E[C_{k,k'}]$ is defined. For most graph spaces, an analytical expression for $\mathbb E[C_{k,k'}]$ is unknown, but by using the MCMC techniques of Section~\ref{sec:sampling}, we can estimate $\mathbb E[C_{k,k'}]$ for any graph space discussed in this paper. Specifically, for each sample graph, and for all degrees $k$ and $k'$ in the degree sequence, we tally the number of edges between vertices of degrees $k$ and $k'$ and then average these counts over all samples to estimate $\mathbb E[C_{k,k'}]$. 

Figure~\ref{fig-disagreeinglocaloptima}(a) shows the non-uniform differences between $\mathbb E[C_{k,k'}]$ for the stub-labeled loopy multigraph and the vertex-labeled multigraph. In particular, edges between vertices with more disparate degrees are more common under the standard stub-labeled loopy multigraph space than the vertex-labeled multigraph space. As a result, the vertex-labeled multigraph modularity function favors grouping connected vertices with differing degrees more than the stub-labeled loopy multigraph modularity function. The vertex-labeled multigraph null model meaningfully changes the landscape of the modularity objective function, which we demonstrate by studying the behavior of two different modularity maximizing algorithms. 

The first algorithm, based on the Kernighan-Lin algorithm, begins with a random partition of the network's vertices into a fixed number of communities. Then, a deterministic local search proceeds by sequentially proposing to move each vertex into each of the other communities. The proposal that most increases or least decreases modularity is accepted and a single full iteration is completed when every vertex has been forced to moved exactly once. The highest modularity partition from one iteration is then used to seed the next iteration, and the algorithm exits when a full iteration passes with no improvement. 

For our investigation we recorded the final partition returned by the algorithm for $K$ communities, where $K=2, 3, \dots, 10$, beginning from $100$ random initial partitions and using Eq.~\eqref{q_newman} as the objective function. Next, starting from the same $100$ initial partitions, we recorded the final partitions using Eq.~\eqref{q_ckk} as the objective function. The two objective functions produced different final partitions from the same initial partitions in a vast majority of cases for $K>2$, as shown in Figure~\ref{fig-disagreeinglocaloptima}(b), and these differences were substantial, as indicated by a normalized mutual information in Figure~\ref{fig-disagreeinglocaloptima}(c) substantially below one. Additionally, we tested whether the locally maximum modularity partitions of one objective function's were also local maxima of the other function, and found that between $9\%$ and $19\%$ were not, indicating that the two null models are in disagreement about the locations of locally optimal partitions.

\begin{figure}
	\centering
	\includegraphics[width=0.32\linewidth]{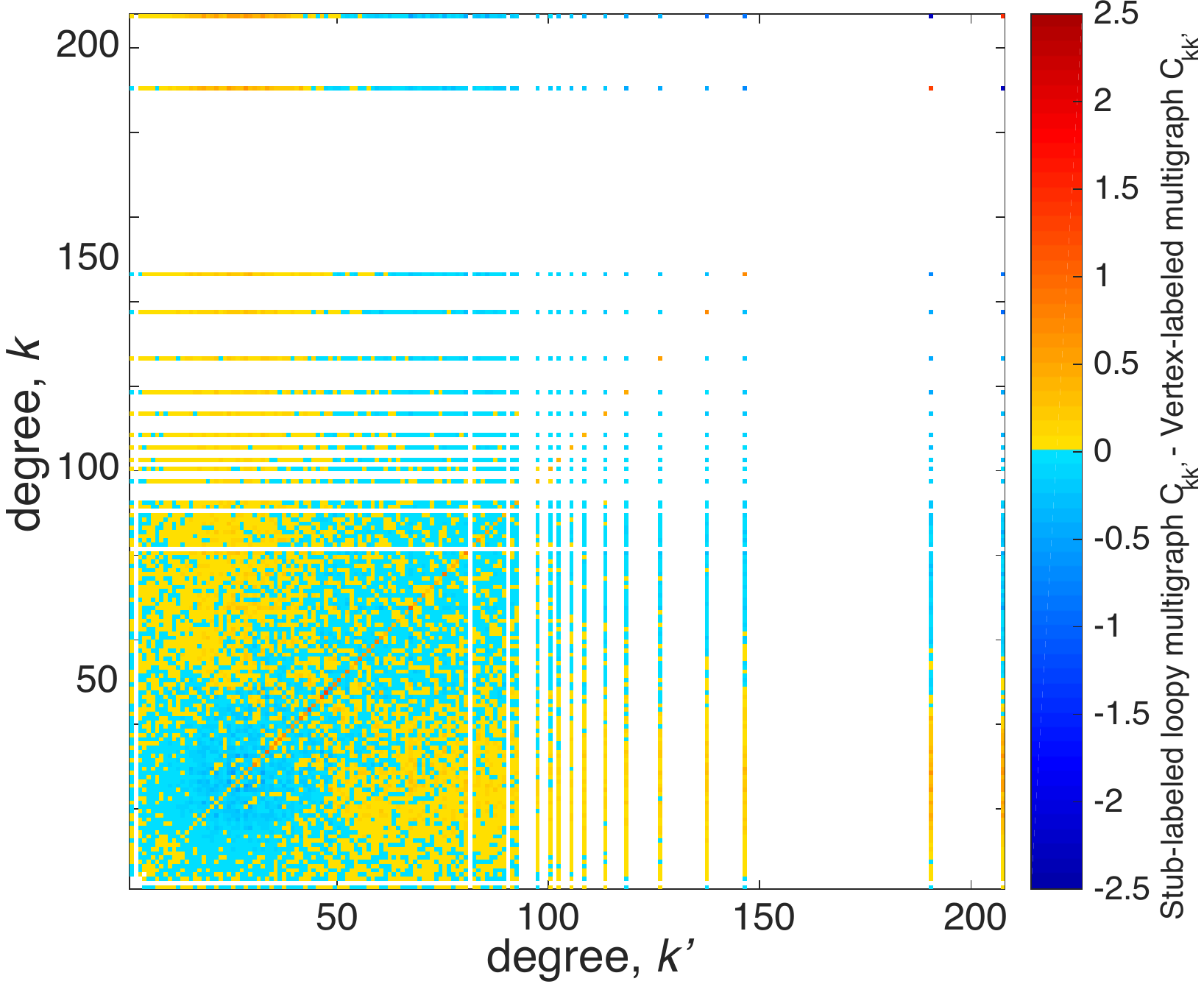}\ 
	\includegraphics[width=0.33\linewidth]{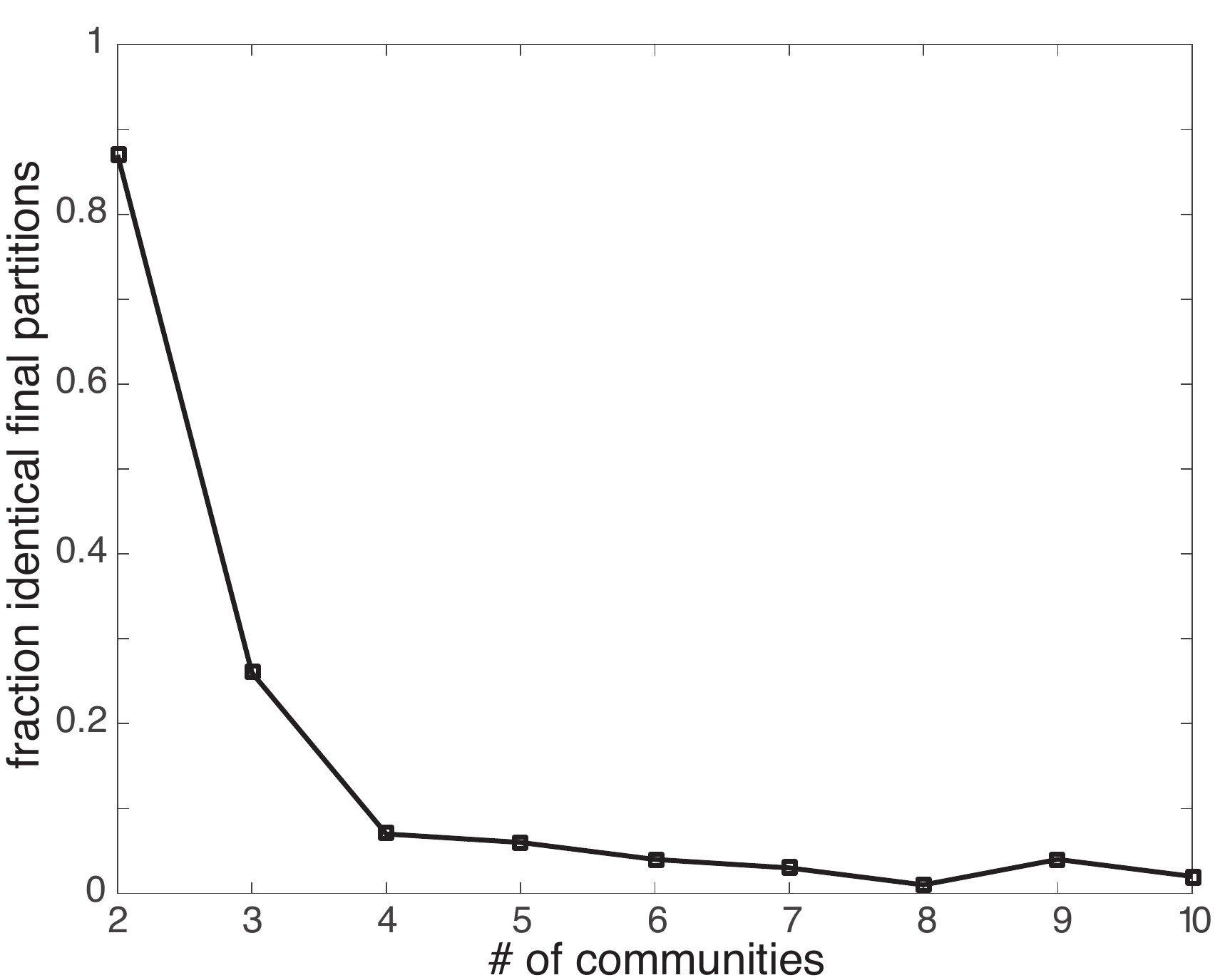}\ 
	\includegraphics[width=0.32\linewidth]{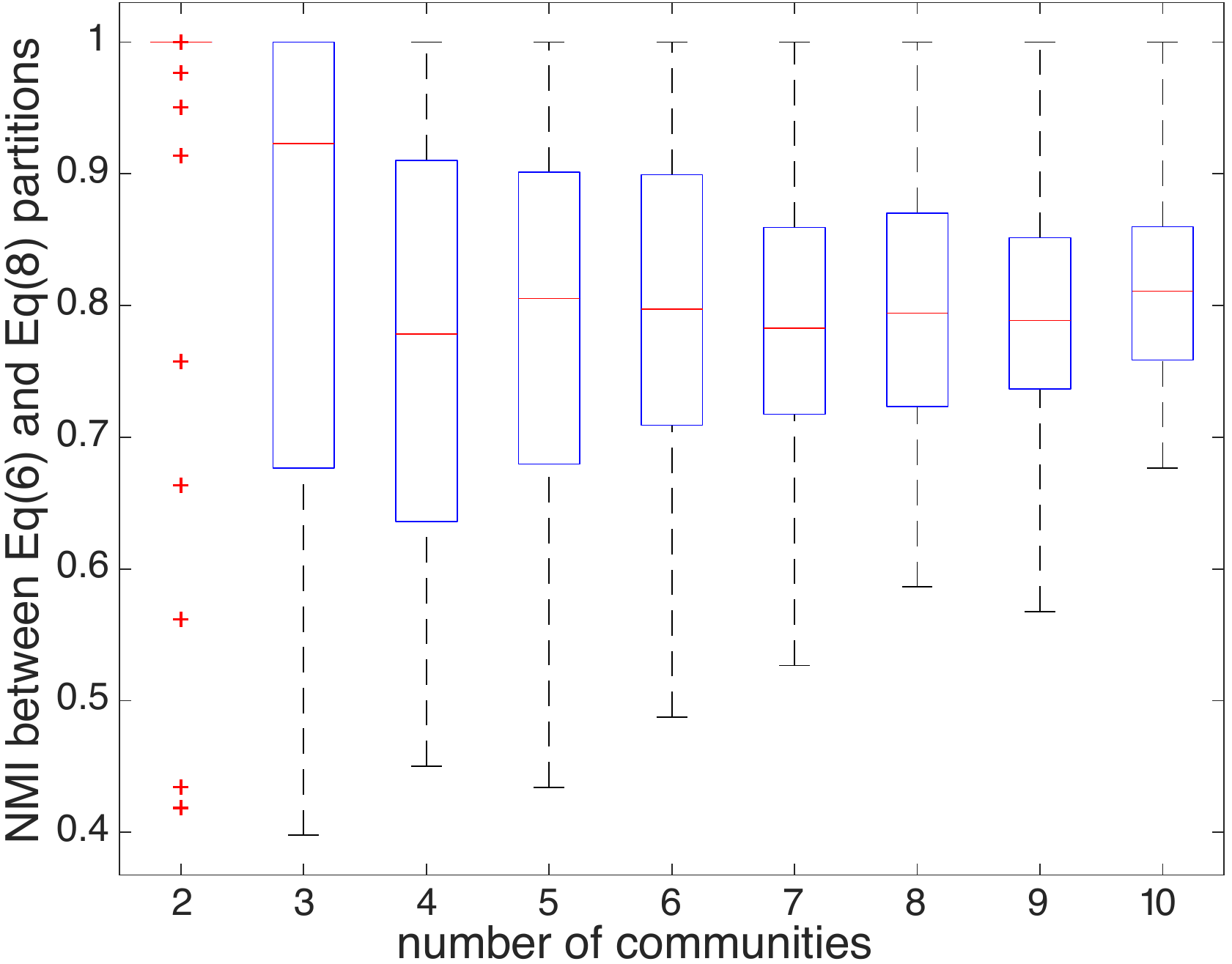}
	\caption{{\bf Choice of configuration model space impacts modularity maximization}. 
	Modularity maximization identified communities in a social multigraph of $782$ vertices under two configuration null models: stub-labeled loopy multigraphs [Eq.~\eqref{q_newman}] and vertex-labeled multigraphs [Eq.~\eqref{q_ckk}].
	(a) Non-uniform differences between the two null model matrices are colored as indicated; white space indicates that there were zero vertices of degree $k$.
	(b) Results of modularity maximization by deterministic local search (see text), starting from identical initial state but using the two different null models, differ for the vast majority of initial states and number of communities $K$, except the case of $K=2$ communities for which $88\%$ agreed. 
	(c) Distributions of normalized mutual information (NMI), which measures similarity of partitions, show differences between the partitions found using the two null models. }
	\label{fig-disagreeinglocaloptima}
\end{figure}

The second modularity maximization algorithm considered is a fully deterministic greedy algorithm that begins with each vertex in a community of its own. Then, at each step of the algorithm there is a proposal to merge every possible pair of communities, and the merger that most increases modularity is chosen \cite{newman2004finding}. This process is repeated sequentially until the vertices are all merged into a single community. From the resulting sequence of partitions and modularity values, we may either select the partition with the highest modularity score or select the partition with a desired number of communities. 

In our investigation of the village social network using this greedy algorithm, the highest modularity partition using both Eqs.~\eqref{q_newman} and \eqref{q_ckk} had 10 communities. These two maximum modularity partitions were identical for both null models, yet partitions were identical at only 43\% (335 of 782) of the agglomerative steps. Thus, while a majority of the algorithm's agglomerative choices differed by null model, a large fraction of partitions remained the same, and both models produced the same optimal partition with 10 communities.

Together, these tests show that modularity maximization, a community detection method based on a configuration model, is sensitive to the particular configuration model used in Eq.~\eqref{q_ckk}. While other algorithms for modularity-based community detection may explore the modularity surfaces using different means, the surfaces themselves are nevertheless distinct. In order to preserve interpretability of modularity maximization's results, practitioners should choose the correct graph space from which the observed network is plausibly drawn.

\section{Conclusions}

Random graphs with fixed degree sequences appear across an enormous number of mathematical and scientific domains, and, until this point, uniform distributions of such graphs have commonly been called {\it the} configuration model. In this paper we showed that the concept of a random graph with a fixed degree sequence can be applied to eight overlapping, yet often meaningfully different graph spaces. We introduced three questions in Section~\ref{subsec:choosing} regarding the presence or absence of self-loops, multiedges, and stub labels, which can be used along with contextual knowledge of a real-world network to decide upon the most appropriate graph space. 

Three applications in Section~\ref{sec:vignettes} highlighted the particularly important distinction between stub-labeled and vertex-labeled spaces. In particular, the use of a stub-labeled configuration model in place of its vertex-labeled counterpart inverted the conclusions of degree-correlation hypothesis tests and changed the optimization landscape for community detection. Simply put, stub- and vertex-labeled spaces are not interchangeable. Simple and non-simple configuration models are not interchangeable either. Although there are widely known asymptotic conditions under which the space of stub-labeled loopy multigraphs contains few graphs with self-loops or multiedges \cite{molloy1995critical}, many graphs analyzed in practical contexts are simply too small or too dense to lean on these asymptotic results. 

As part of our work, we presented three Markov chain Monte Carlo sampling algorithms and proved that they can be used to generate graphs uniformly from the eight graph spaces discussed. To that end, pseudocode and Python implementations are provided, as used in the three applications of Section~\ref{sec:vignettes}. However, as with most algorithms, there are tradeoffs. While these MCMC approaches are proven to uniformly sample from the desired graph space, rigorous mixing time bounds have not been established, and we look forward to future mixing time investigations. 

 Throughout this paper we discussed and reviewed the wide and disparate history of configuration models and their sampling techniques, drawing on literature from sociology, ecology, combinatorics, statistics, and physics. Many results regarding configuration models have been discovered multiple times, in part due to the deep and scattered literature, and in part due to the fact that that there exist various names given to the same set of models, and one name given to multiple different models. It is therefore our hope that the results and summaries in this paper help to clarify and refine the study of configuration models, their graph spaces, and their applications. 

\thanks{{\bf Acknowledgements.} The authors would like to thank Kristen Altenberger, Aaron Clauset, Iris Levin, Stephen Ragain, Henry Scharf, Alice Schwarze and the Networks Journal Club at University of Oxford for thoughtful comments on early versions of the manuscript, Eleanor Power for providing South India social support network data, Iris Levin and Rebecca Safran for providing data on barn swallow contacts, and the American Mathematical Society's Mathematics Research Community on Networks for bringing the authors together.}

\def\refname{References}
\bibliography{references}{}
\bibliographystyle{plain}

\clearpage
\appendix

\section*{Supplementary Materials: Algorithm~\ref{alg_vertex_complex}} 
\small
This algorithm uniformly samples vertex-labeled graph spaces more efficiently than Algorithm~\ref{alg_vertex_basic} by computing both the forward and reverse probabilities of any double-edge swap according to the cases in Figure \ref{EdgeSwapCases}. It then down-samples the higher probability swap to have the same probability as the lower probability swap, accelerating mixing.

\begin{algorithm}
\caption{vertex-labeled MCMC} \label{alg_vertex_complex} 
\end{algorithm}
\begin{algorithmic}
\REQUIRE {initial graph $G_0$, graph space (simple graph, multigraph, or loopy multigraph)}
\ENSURE{sequence of graphs $G_i$}
\FOR {$i<$ number of graphs to sample} 
\STATE choose two distinct edges $(u,v)$ and $(x,y)$ uniformly at random
\IF {$Unif(0,1)<0.5$ }
	\STATE $u,v\leftarrow v,u$ 
\ENDIF
\IF {edge swap would leave the graph space} 
	\STATE resample current graph: $G_{i} \leftarrow G_{i-1}$
\ENDIF
\IF {$\exists$ 4 distinct vertices in $u,v,x,y$}
	\STATE $SwapsTo\leftarrow w_{uv}w_{xy}$ 
	\STATE $SwapsFrom \leftarrow (w_{ux}+1)(w_{vy}+1)$ 
\ELSIF {$\exists$ 3 distinct vertices in $u,v,x,y$}
	\IF {$u=v$ or $x=y$}
    		\STATE $SwapsTo \leftarrow 2w_{uv}w_{xy}$ 
    		\STATE $SwapsFrom \leftarrow (w_{ux}+1)(w_{vy}+1)$ 
	\ELSE
    		\STATE $SwapsTo \leftarrow w_{uv}w_{xy}$ 
    		\STATE $SwapsFrom \leftarrow 2(w_{ux}+1)(w_{vy}+1)$ 
	\ENDIF
\ELSIF {$\exists$ 2 distinct vertices in $u,v,x,y$}
	\IF {only one of $(u,v)$ or $(x,y)$ is a self-loop}
		\STATE $G_{i} \leftarrow G_{i-1}$ 
		\STATE continue 
	\ELSIF{both $(u,v)$ and $(x,y)$ are self-loops}
		\STATE $SwapsTo\leftarrow 2w_{uu}w_{xx}$ 
    		\STATE $SwapsFrom \leftarrow \frac{1}{2}(w_{ux}+2)(w_{ux}+1)$ 
	\ELSE
		\STATE $SwapsTo\leftarrow \frac{1}{2}w_{uv}(w_{uv}-1)$ 
		\STATE $SwapsFrom \leftarrow 2(w_{uu}+1)(w_{vv}+1)$ 
	\ENDIF
\ELSE
	\STATE $G_{i} \leftarrow G_{i-1}$ 
	\STATE continue 
\ENDIF
\STATE $P\leftarrow \min (1,\frac{SwapsFrom }{SwapsTo})$
\IF {$Unif(0,1)<P$}
	\STATE {swap $(u,v),(x,y)\leadsto(u,x),(v,y) $ to produce $G_{i}$}
\ELSE
	\STATE $G_{i} \leftarrow G_{i-1}$ 
\ENDIF
\ENDFOR
\end{algorithmic}

\clearpage

\begin{center}
\begin{figure}[h]
	\centering
	\includegraphics[width=\linewidth]{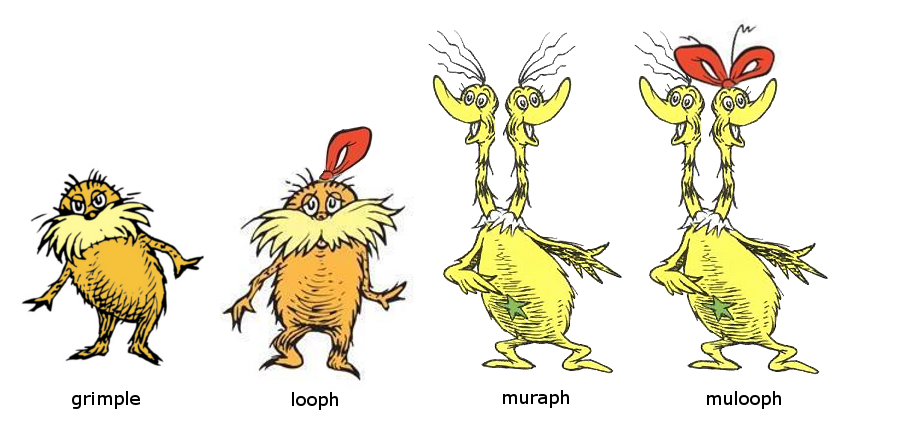}
\end{figure}
Four species of graphs lived on the last page,\\
The wise old mulooph was the most cited sage.\\
$ $

To her left, the muraph, no loop-looping bow,\\
The second key graph in the modeling show.\\
$ $

Yet all was not well with the old grumpy grimple,\\
Since basic stub matching just can't make graphs simple.\\
$ $

Last graph of the four, an elusive young looph,\\
Whose counterexample demands a new proof.\\
$ $

Each graph in its space, with uniform frequence,\\
A place for each case---and with fixed degree sequence!\\

$$ $$
``Don't cry because it's over. Smile because it happened." -- Dr. Seuss
\end{center}

\end{document}